\documentclass[12pt]{article}

\usepackage[paper=letterpaper,margin=1in]{geometry}

\title{
Egalitarian Resource Sharing Over Multiple Rounds\footnote{	Department of Computer Science,
	University of Texas at Austin.  Email: \{fuli, plaxton, vaibhavsinha\}@utexas.edu.
}
}

\author{Fu Li \and  C.\ Gregory Plaxton \and Vaibhav B.\ Sinha}

\date{\today}

\usepackage{amsthm}

\usepackage{amsmath,amsfonts,amssymb}
\usepackage{mathtools}
\usepackage{algorithm,algpseudocode}
\usepackage{hyperref}
\usepackage{graphicx}
\usepackage{color}
\usepackage[dvipsnames]{xcolor}
\usepackage{bibunits}

\newtheorem{theorem}{Theorem}

\newtheorem{lemma}{Lemma}

\newtheorem{ob}{Observation}
\newtheorem{corollary}{Corollary}

\newcommand{\punt}[1]{}
\newcommand{\tmpDelete}[1]{}

\newcommand{\CoeffSI}{z}

\newcommand{\agents}{A}
\newcommand{\agent}{a}
\newcommand{\objects}{B}
\newcommand{\object}{b}
\newcommand{\edge}{e}
\newcommand{\edges}{E}
\newcommand{\edw}{\alpha}
\newcommand{\supp}{\beta}
\newcommand{\capacity}{C}
\newcommand{\capt}{c}
\newcommand{\dem}{d}
\newcommand{\utility}{u}
\newcommand{\ut}[3]{\utility(#1, #2, #3)}
\newcommand{\utobj}[4]{\utility(#1, #2, #3, #4)}

\DeclareMathOperator{\oafd}{OAFD}
\DeclareMathOperator{\sort}{sort}
\DeclareMathOperator{\suppl}{supplies}
\DeclareMathOperator{\edows}{endowments}
\DeclareMathOperator{\demands}{demands}
\DeclareMathOperator{\allo}{allocs}
\DeclareMathOperator{\nmbr}{num}
\DeclareMathOperator{\agentss}{agents}
\DeclareMathOperator{\objectss}{objects}
\DeclareMathOperator{\caap}{cap}
\DeclareMathOperator{\brkptss}{brkpts}
\DeclareMathOperator{\suub}{sub}
\DeclareMathOperator{\LMMF}{LMMF}
\DeclareMathOperator{\MMF}{MMF}
\DeclareMathOperator{\frugal}{frugal}
\DeclareMathOperator{\shrink}{shrink}

\newcommand{\allocationSymbol}{\mu}
\newcommand{\supply}[1]{\suppl(#1)}
\newcommand{\edws}[1]{\edows(#1)}
\newcommand{\dems}[2]{\demands(#1, #2)}
\newcommand{\allocs}[1]{\allo(#1)}
\newcommand{\num}[1]{\nmbr(#1)}
\newcommand{\brkfn}[2]{\Lambda_{#1}(#2)}
\newcommand{\brkpts}[2]{\brkptss(#1, #2)}
\newcommand{\agnts}[2]{\agentss(#1, #2)}
\newcommand{\objs}[2]{\objectss(#1, #2)}
\newcommand{\suppt}[2]{\supp_{#1}(#2)}
\newcommand{\capty}[2]{\caap(#1, #2)}
\newcommand{\brkpt}{\lambda}
\newcommand{\newsub}[3]{\suub(#1, #3, #2)}
\newcommand{\Manip}{\Phi}
\newcommand{\LMM}{\mathcal{M}}

\newcommand{\AAA}{\mathcal{A}}
\newcommand{\SI}{ \text{SI}}
\newcommand{\st}[1]{\sort(#1)}

\newcommand{\Winners}[1]{W(#1)}
\newcommand{\Losers}[1]{L(#1)}

\begin{document}
\maketitle

\begin{abstract}
	It is often beneficial for agents to pool their resources in order to
better accommodate fluctuations in individual demand.  Many
multi-round resource allocation mechanisms operate in an online
manner: in each round, the agents specify their demands for that
round, and the mechanism determines a corresponding allocation.  In
this paper, we focus instead on the offline setting in which the
agents specify their demand for each round at the outset.  We
formulate a specific resource allocation problem in this setting, and
design and analyze an associated mechanism based on the solution
concept of lexicographic maximin fairness. We present an efficient
implementation of our mechanism, and prove that it is
\tmpDelete{Pareto-efficient,} envy-free, non-wasteful, resource monotonic, population monotonic,  and
group strategyproof. We also prove that our mechanism guarantees each
agent at least half of the utility that they can obtain by not sharing
their resources.  We complement these positive results by proving that
no maximin fair mechanism can improve on the aforementioned factor of
one-half.

\end{abstract}

\section{Introduction}
\label{sec:intro}

 
Problems related to computational resource allocation lie at the
intersection of economics and computer science, and have received a
lot of attention in the literature.  In particular, the theory of fair
division, including such concepts as the egalitarian equivalent rule,
provides a suitable game-theoretic framework for tackling modern technological
challenges arising in cloud computing environments.  
This connection
has inspired several mechanisms with strong game-theoretic properties,
including
mechanisms for coping with fluctuating
demands~\cite{cole2013mechanism,freeman2018dynamic,tang2014long}, and
for allocating multiple resource types such as CPU, disk, and
bandwidth~\cite{ghodsiZHKSS2010,ghodsi2013choosy} when agents do not
know their resource demands~\cite{kandasamy2020online} or when there
is a stream of resources~\cite{aleksandrov2017pure}.

In this paper, we consider a group of agents sharing a single type of
resource (e.g., a collection of identically-equipped servers in the
cloud) over a set of rounds. (Later we allow for the possibility that
the total supply of the shared resources may vary from one round to
the next.)
Each agent owns a specific fraction of the shared resources, and
reports a demand for each round.  For a given round, an agent accrues
utility equal to the (possibly fractional) number of units allocated
to them, as long as the allocation does not exceed their demand; any
allocation beyond this threshold does not provide additional
utility. We assume that no monetary exchange occurs between the
agents, as is often the case in resource sharing applications (e.g.,
within a single organization).  Even in this simple setting, there are
many interesting questions that can be investigated. Indeed, a number
of works (e.g.,
\cite{freeman2018dynamic,hossain19,kandasamy2020online,tang2014long})
have proposed and analyzed allocation mechanisms for this setting.
These works have emphasized the natural online variant in which the
agent demands for any given round are not revealed until the start of
that round. In some applications, it may be possible to accurately
estimate future demands, e.g., due to periodicity. Can we design
mechanisms that effectively exploit such (partial) knowledge of future
demands? From a theoretical perspective, a natural starting point for
addressing this question is to consider the offline variant in which
all of the future demands are known at the outset; this is the
approach taken in the present paper.

Strategic (coalitions of) agents might misreport their demands to try
to achieve higher utility, often at the expense of other agents.  For
example, an agent might under-report their demand in a given round,
hoping for any loss of utility in that round to be more than offset by
the net gain realized in the remaining rounds.  Accordingly, we seek
to design strategyproof (SP) or group strategyproof (GSP) mechanisms
that incentivize truthful reporting.

An allocation satisfies the sharing incentives (SI) property if it
ensures that each agent achieves utility at least as high as they
can achieve by not sharing their resources.
The mechanism that allocates resources to the agents in proportion to
their relative endowments (ownership shares) is GSP and SI.  Such a
mechanism can be wasteful in the sense that it can allocate resources to an agent
in excess of their demand while leaving the demand of another agent unmet.  
Thus, we seek to design mechanisms that only produce non-wasteful (NW) allocations.

Given an allocation, we say that agent $a$ envies agent $a'$ if $a$
prefers the allocation of $a'$ (scaled to account for the relative
endowments of $a$ and $a'$) to their own.  
An allocation is envy-free (EF) if no agent envies another.
An allocation is frugal if it does not allocate more resources to an agent than they demand.
We consider two notions of fair allocations:
lexicographically maximin fair (LMMF),
and a weaker notion, maximin fair (MMF).

We seek mechanisms that are resource monotonic (RM),
that is, if the supply of one or more resources is increased, no agent
experiences a decrease in utility.
We seek mechanisms that are also population monotonic (PM),
that is, if the  endowment of one or more agents is decreased, no other agent
experiences a decrease in utility.


For the offline setting, we present an egalitarian mechanism for allocating
resources to agents over multiple rounds and provide an efficient
algorithm to compute the allocation.  Our mechanism is frugal,
LMMF, GSP, $\frac12$-SI (a relaxation of SI), NW, 
EF, RM,
and PM.  We also show that there is no MMF $\CoeffSI$-SI mechanism for
any $\CoeffSI > \frac12$.  The significance of our work is discussed in
greater detail below, but first we review relevant prior work.

{\bf Related work.} 
A lexicographic maximin solution maximizes the minimum utility, and
subject to this, maximizes the second-lowest utility, and subject to
this, maximizes the third-lowest utility, and so on.
Lexicographic maximin solutions have been studied in many area of
research, including computing the nucleolus of cooperative
games~\cite{pottersT92}, combinatorial
optimization~\cite{behringer1981simplex}, network
flows~\cite{megiddo1974, megiddo1977}, and as one of the standard
fairness concepts in telecommunications and network
applications~\cite{ogryczak2005telecommunications, pioro2004routing}. For
more details, we refer the reader to the recent work of Ogryczak et
al.~\cite[Section~2.1]{ogryczakLNP14_survey_fair}. 
Below we briefly discuss the  works that are most relevant to the present paper.

Some widely used online schedulers (e.g., the fair scheduler
implemented in Hadoop and Spark) enforce LMMF.  In the online setting,
there are two senses in which we can seek to achieve the LMMF
property: static and dynamic.  In the static sense, we produce an LMMF
allocation for each round independently.  In the dynamic sense, we
produce an allocation for any given round that enforces LMMF over the
entire history up to that round (subject to the constraint that the
allocations determined for previous rounds cannot be changed).

Our work is inspired by Freeman et al.~\cite{freeman2018dynamic}, who
studied the game-theoretic aspects of online resource sharing, with a
primary focus on the SP, SI, and NW properties.  They prove that the static version discussed above 
satisfies these desiderata, while the dynamic version 
fails to satisfy SI and SP.
They then consider a more general utility function, where agents
derive a fixed ``high'' utility per unit of resource up to their
demands and a fixed ``low'' utility beyond that threshold. With this
utility function, Freeman et al.\ show that the three aforementioned
properties are incompatible in a dynamic setting and thus appropriate
trade-offs need to be considered. They propose two mechanisms that
partly satisfy the desiderata. Hossain~\cite{hossain19} has
subsequently presented another mechanism for this setting.

Kandasamy et al.~\cite{kandasamy2020online} study mechanism design for
online resource sharing when agents do not know their resource
requirements.  Like Freeman et al., they focus on satisfying the SP,
SI, and NW properties.  Tang et al.\ \cite{tang2014long} propose a
dynamic allocation policy for the online setting that is similar to
the dynamic version of LMMF.

Bogomolnaia and Moulin~\cite{bogomolnaia2004random} study random
assignment problems with dichotomous preferences from a
game-theoretical perspective. Dichotomous preferences can be viewed as
a special case of fractional demands.  Bogomolnaia and Moulin consider
several mechanisms, including the LMMF mechanism.  They prove that the
latter mechanism is GSP, EF, RM, PM, and fair-share (the special case
of SI where all demands are zero or infinite).  Compared with the
model of Bogomolnaia and Moulin, our setting allows for fractional
demands, unequal agent endowments, and an unequal supply of resources
between rounds.

The work of Katta and Sethuraman~\cite{katta2006solution}
addresses
the random assignment problem with general agent preferences (i.e.,
where indifference is allowed in the agent preferences).
They use parametric network flow to achieve LMMF for the special case
of dichotomous preferences.  For general preferences, they extend the
parametric flow algorithm to compute an EF and ordinally efficient
assignment, and they prove that no mechanism is SP, EF, and ordinally
efficient.

Ghodsi et al.~\cite{ghodsi2013choosy} 
consider the LMMF mechanism 
in the context of a random assignment problem where
the agent endowments need not be the same.  
Our work strengthens  their SP result to GSP while allowing for fractional demands.

Offline resource sharing has been studied in the context of
multiperiod resource allocation with equal agent
endowments~\cite{klein1992lexicographic,luss1999equitable}. 
This line
of research is focused on the design of efficient algorithms for
computing a lexicographic maximin solution (via linear programming), as
opposed to analyzing the associated game-theoretic properties.

Offline resource sharing can be viewed as the problem of allocating
different kinds of substitutable resources to different populations of
agents.  This allocation problem has been studied in various specific
settings, e.g., distribution of coal among power
companies~\cite{brown1979sharing},
 multiperiod manufacturing of
high-tech products~\cite{klein1995multiperiod}, and allocation of
vaccines to different populations~\cite{singh2020fairness}.
To
the best of our knowledge, the prior work in this area studies this
allocation problem from a computational perspective, rather than a
game-theoretic perspective.  Sethuraman's survey paper on house
allocation problems~\cite{sethuraman2010_survey_house_allocation}
discusses the connection between allocation with substitutable
resources and random assignment with dichotomous preferences.


{\bf Significance of our work.} Freeman et
al.~\cite{freeman2018dynamic} study the game-theoretic properties of
several online resource allocation mechanisms: the previously known
static and dynamic LMMF mechanisms, and the newly-proposed Flexible
Lending and T-period mechanisms~\cite{freeman2018dynamic}. In settings
where future demands are known, or can be accurately estimated, we can
hope to significantly improve upon the fairness guarantees of such
online mechanisms.  As a simple example, consider an instance with $n$
agents and $n$ rounds, where each agent contributes a single unit per
round.  Suppose agents $2$ through $n$ each demand two units in every
round, and agent $1$ demands $n$ units in round $1$ and no units
thereafter.  Clearly, an egalitarian allocation gives a utility of $n$
to every agent. On the other hand, all of the aforementioned online
mechanisms give agent $1$ a utility of $1$, which is only a $1/n$
fraction of the egalitarian share. Our first main contribution is to
provide an efficient implementation of a suitable egalitarian
mechanism for the offline setting (i.e., where future demands are
known); see the first part of Section \ref{sec:mechanism}.

Our second main contribution is to establish various fundamental
game-theoretic properties of our egalitarian mechanism. To do so, we
leverage a connection between the random assignment problem of
Bogomolnaia and Moulin~\cite{bogomolnaia2004random} and our resource
sharing problem. Specifically, they assume that the agent preferences
are dichotomous, which corresponds to the special case of our setting
in which the fractional demands of the agents are all $0$ or $1$.
While the work of Bogomolnaia and Moulin provides us with an
invaluable roadmap, we need to overcome some technical challenges in
order to handle arbitrary fractional demands. (We also handle unequal
agent endowments and unequal supplies over the rounds, but
generalizing our results in these directions is quite
straightforward.) In Section~\ref{sec:structural}, we establish a
number of useful structural properties of lexicographic maximin
allocations. In Section~\ref{sec:allocProps}, we use these structural
properties to establish various game-theoretic properties of frugal
LMMF allocations.  In Section~\ref{sec:GSP}, we establish that our
egalitarian mechanism (and in fact any frugal LMMF mechanism) is GSP.

Our third main contribution is to establish possibility and
impossibility results related to the SI property.  The SI property is of particular importance in the setting of resource sharing, 
where we need to ensure that agents are  not discouraged from  pooling their resources. 
In Section \ref{sec:allocProps}, we show that any frugal LMMF allocation is $\frac{1}{2}$-SI.  In Section \ref{sec:imp}, we
 show that, for any $\CoeffSI>\frac{1}{2}$, no mechanism is MMF and
$\CoeffSI$-SI. Since no mechanism is MMF and SI, we consider a natural
relaxation: mechanisms that are MMF subject to being SI. (In other
words, we require the mechanism to be SI, and we only enforce the MMF property with
respect to the set of SI allocations.)  In Section \ref{sec:imp}, we show that no such mechanism
is SP.
All of our proofs appear in the appendix.

\section{Preliminaries}
\label{sec:prelims}

As discussed in the introduction, we wish to model a setting in
which a number of agents share a set of resources over multiple
rounds, and where the demand of each agent in each round is known in
advance. For the purposes of our formal presentation, we find it
convenient to refer to the pool of available resources in a given
round as an object.  (We will allow the size of this pool to vary
from one round to the next; see the notion of ``supply'' defined
below.) Thus, we use $k$ objects to model a $k$-round instance.
Since we are restricting attention to the offline setting, the
ordering of these objects is immaterial.

For any set of agents $\agents$, we define $\edws{\agents}$ as the set
of all endowment functions $\edw : \agents \to \mathbb{R}_{> 0}$, and
for any subset $\agents'$ of $\agents$ we define $\edw(\agents')$ as
$\sum_{\agent \in \agents'} \edw(\agent)$.\footnote{In the remainder
of the paper, we implicitly define similar overloads for a number of
other functions associated with supply, demand, allocation, flow, 
and capacity.}  For any set of objects $\objects$, we define
$\supply{\objects}$ as the set of all supply functions $\supp :
\objects \to \mathbb{R}_{\ge 0}$. 

For any set of agents $\agents$ and any set of objects $\objects$, 
we define $\dems{\agents}{\objects}$ as the set of all demand  functions $\dem : \agents \times \objects \to \mathbb{R}_{\ge 0}$. 
For any subset $\agents'$ of $\agents$ and any $\dem$ in $\dems{\agents}{\objects}$, $\dem_{\agents'}$ denotes the demand function in $\dems{\agents'}{\objects}$ such that $\dem_{\agents'}(\agent, \object) = \dem(\agent, \object)$ for all agents $\agent$ in $\agents'$ and all objects $\object$ in $\objects$.

For any set of agents $\agents$, any set of objects $\objects$, any $\edw$ in $\edws{\agents}$, any $\supp$ in $\supply{\objects}$, and any $\dem$ in $\dems{\agents}{\objects}$, 
the tuple $(\agents, \objects, \edw, \supp, \dem)$ denotes
an instance of object allocation with fractional demands 
($\oafd$).
We think of each agent $\agent$ in $\agents$ as owning a $\edw(\agent)/\edw(\agents)$
fraction of each object in $\objects$.

For any $\oafd$ instance $I = (\agents, \objects, \edw, \supp, \dem)$, 
we define $\allocs{I}$ as the set of all allocation functions  $\allocationSymbol : \agents \times \objects \to \mathbb{R}_{\ge 0}$ such that 
$\allocationSymbol(\agents,\object) \le \supp(\object)$ for all objects $\object$ in $\objects$.
 
An $\oafd$ mechanism $M$ takes as input an $\oafd$ instance 
$I$ and  
outputs a subset $M(I)$ of $\allocs{I}$.\footnote{In the present paper, it is convenient to assume that the output of an $\oafd$ mechanism is a set of allocations, as opposed to a single allocation, because the $\oafd$ mechanism $\LMM$ that we present in Section~\ref{sec:mechanism} has this characteristic.
For any $\oafd$ instance $I = (\agents, \objects, \edw, \supp, \dem)$, all of the agents in $\agents$ are indifferent between the allocations in $\LMM(I)$. For any $\oafd$ instance $I$, our efficient implementation of $\LMM$ computes a single allocation in $\LMM(I)$.}

For any $\oafd$ instance $I = (\agents, \objects, \edw, \supp, \dem)$,
and any $\allocationSymbol$ in $\allocs{I}$, we define the utility of
agent $\agent$ from object $\object$ as
$\utobj{\allocationSymbol}{\dem}{\agent}{\object} =
\min(\allocationSymbol(\agent, \object), \dem(\agent, \object))$.  We
assume that the utility of any agent $\agent$, denoted
$\ut{\allocationSymbol}{\dem}{\agent}$, is equal to $\sum_{\object \in
  \objects} \utobj{\allocationSymbol}{\dem}{\agent}{\object}$.

For any $\oafd$ instance $I = (\agents, \objects, \edw, \supp, \dem)$,
any $\allocationSymbol$ in $\allocs{I}$, and any $\object$ in
$\objects$, the definition of utility implies that $\sum_{\agent \in
  \agents} \utobj{\allocationSymbol}{\dem}{\agent}{\object} \le
\min(\supp(\object), \dem(\agents, \object))$.  We let $\supp_I$ in
$\supply{I}$ denote the supply such that $\suppt{I}{\object} =
\min(\supp(\object), \dem(\agents, \object))$ for all objects
$\object$ in $\objects$.  For any $\oafd$ instance $I = (\agents,
\objects, \edw, \supp, \dem)$, any $\allocationSymbol$ in
$\allocs{I}$, and any subset $\agents'$ of $\agents$, the definition
of utility also implies that $\sum_{\agent \in \agents'}
\ut{\allocationSymbol}{\dem}{\agent} \le \sum_{\object \in \objects}
\min(\suppt{I}{\object}, \dem(\agents', \object))$.  We let
$\capty{I}{\agents'}$ denote $\sum_{\object \in \objects}
\min(\suppt{I}{\object},\dem(\agents', \object))$.

{\bf Game-theoretic desiderata for allocations.}
For any $\oafd$ instance $I = (\agents,
\objects, \edw, \supp, \dem)$, an allocation $\allocationSymbol$ in
$\allocs{I}$ is (proportionally) envy-free (EF) if
\begin{equation*}
	\utility(\allocationSymbol,\dem, a) \ge \sum_{\object \in \objects}
\min\left(\frac{\edw(\agent)}{\edw(\agent')}
\allocationSymbol(\agent',\object),\dem(\agent,\object)\right)
\end{equation*}
for
all agents $\agent$ and $\agent'$ in $\agents$.  Intuitively, no agent
prefers the appropriately scaled (i.e., taking into account relative
endowments) version of another agent's allocation to their own
allocation.


The sharing incentives (SI) property requires that any agent 
$\agent$ who provides a
truthful report achieves utility at least as high as they would
achieve with an $\edw(\agent)/\edw(\agents)$ fraction of every object.
Formally, for any $\oafd$ instance
$I=(\agents,\objects,\edw,\supp,\dem)$ and any $\CoeffSI$ in $[0,1]$, an
allocation $\allocationSymbol$ in $\allocs{I}$ is said to be
$\CoeffSI$-SI if $\ut{\allocationSymbol}{\dem}{\agent} \ge \CoeffSI\sum_{\object \in
  \objects} \min\left(
\frac{\edw(\agent)}{\edw(\agents)}\supp(\object), \dem(\agent,
\object)\right)$ for all agents $\agent$ in $\agents$.  We say that an
allocation is SI if it is $1$-SI.

In our model, the maximum utility that an agent $\agent$ can achieve
from an object $\object$ is $\dem(\agent,\object)$; accordingly, in our setting, there
is no reason to allocate more than $\dem(\agent,\object)$ units of
object $\object$ to agent $\agent$.  For any $\oafd$ instance
$I=(\agents,\objects,\edw,\supp,\dem)$, we say that an allocation
$\allocationSymbol$ in $\allocs{I}$ is frugal if
$\allocationSymbol(\agent, \object) \le \dem(\agent, \object)$ for all
$(\agent,\object)$ in $\agents\times\objects$.  We let $\frugal(I)$
denote the set of all frugal allocations in $\allocs{I}$.  
For any $\oafd$ instance $I = (\agents, \objects, \edw, \supp, \dem)$ 
and any $\allocationSymbol$ in $\frugal(I)$, we say that $\allocationSymbol$ is 
non-wasteful (NW) if for any object $\object$ in $\objects$, 
either $\allocationSymbol(\agents, \object) = \supp(\object)$ or $\allocationSymbol(\agent, \object) = \dem(\agent, \object)$ for all agents $\agent$ in $\agents$. 

An
allocation $\allocationSymbol$ in $\allocs{I}$ 
that maximizes $\min_{\agent \in \agents}
\ut{\allocationSymbol'}{\dem}{\agent}/\edw(\agent)$ over all
$\allocationSymbol'$ in $\allocs{I}$
is said to be
maximin fair (MMF).  We let $\MMF(I)$ denote the set
of all MMF allocations in $\allocs{I}$.  We let $\mathbf{\utility}(I,
\allocationSymbol)$ denote the length-$|\agents|$ vector whose $j$th
component denotes the $j$th smallest
$\ut{\allocationSymbol}{\dem}{\agent}/\edw(\agent)$ for all agents
$\agent$ in $\agents$.  An allocation $\allocationSymbol$ in
$\allocs{I}$ is lexicographically maximin fair (LMMF) if $\mathbf{u}(I, \allocationSymbol)$ is
lexicographically at least $\mathbf{u}(I, \allocationSymbol')$ for all
$\allocationSymbol'$ in $\allocs{I}$.  We let $\LMMF(I)$ denote the
set of all LMMF allocations in $\allocs{I}$.  
Note that LMMF is a
stricter notion of fairness than MMF.

For any $\oafd$ instance $I=(\agents,\objects,\edw,\supp,\dem)$, any
subset $\agents'$ of $\agents$, and any $\allocationSymbol$ in
$\LMMF(I)$, we let $\newsub{I}{\allocationSymbol}{\agents'}$ denote
the $\oafd$ instance
$(\agents\setminus\agents',\objects,\edw',\supp',\dem_{\agents
  \setminus \agents'})$ where $\edw'(\agent)=\edw(\agent)$ for all
agents $\agent$ in $\agents \setminus \agents'$ and $\supp'(\object) =
\supp(\object) - \allocationSymbol(\agents', \object)$ for all objects
$\object$ in $\objects$.  Lemma~\ref{lem:suboptimal} below establishes
an optimal substructure property of LMMF allocations.
\begin{lemma}\label{lem:suboptimal}
	Let $I = (\agents, \objects, \edw, \supp, \dem)$ be an $\oafd$ instance, let $\agents'$ be a subset of $\agents$, and let $\allocationSymbol$ belong to $\LMMF(I)$.
	Let $\allocationSymbol'$ be the restriction of $\mu$ to $\agents \setminus A'$, that is, $\allocationSymbol': (\agents \setminus \agents') \times \objects \to \mathbb{R}_{\ge 0}$ is such that $\allocationSymbol'(\agent, \object) = \allocationSymbol(\agent, \object)$ for all $(\agent, \object)$ in $(\agents \setminus \agents') \times \objects$. 
	Then $\allocationSymbol'$ belongs to $\LMMF(\newsub{I}{\allocationSymbol}{\agents'})$. 
\end{lemma}
\punt{\begin{proof}
For any $\oafd$ instance $\hat{I} = (\hat{\agents}, \hat{\objects},
\hat{\edw}, \hat{\supp}, \hat{\dem})$, any $\hat{\allocationSymbol}$
in $\LMMF(\hat{I})$,
and any subset $\hat{\agents}'$ of $\hat{\agents}$, let
$\mathbf{\utility}(\hat{I}, \hat{\allocationSymbol}, \hat{\agents}')$
denote the length-$|\hat{\agents}'|$ vector whose $j$th component
denotes the $j$th smallest value of
$\ut{\hat{\allocationSymbol}}{\hat{\dem}}{\agent}/\edw(\agent)$ over
all agents $\agent$ in $\hat{\agents}'$.  Then
\[
\mathbf{\utility}(\hat{I}, \hat{\allocationSymbol}) =
\st{\mathbf{\utility}(\hat{I}, \hat{\allocationSymbol},
  \hat{\agents}') + \mathbf{\utility}(\hat{I},
  \hat{\allocationSymbol}, \hat{\agents} \setminus \hat{\agents}')},
\]
where $+$ denotes concatenation and $\sort$ is a function that sorts
the input vector.

Let $\tilde{I}$ denote $\newsub{I}{\allocationSymbol}{\agents'}$.
Notice that $\allocationSymbol'$ belongs to $\allocs{\tilde{I}}$.
Assume for the sake of contradiction that $\allocationSymbol'$ does
not belong to $\LMMF(\tilde{I})$.  Let $\tilde{\allocationSymbol}$
belong to $\LMMF(\tilde{I})$.  Then $\mathbf{\utility}(\tilde{I},
\allocationSymbol') = \mathbf{\utility}(I, \allocationSymbol, \agents
\setminus \agents') \neq \mathbf{\utility}(\tilde{I},
\tilde{\allocationSymbol})$.
Since $\tilde{\allocationSymbol}$ belongs to $\LMMF(\tilde{I})$ and
$\allocationSymbol'$ does not belong to $\LMMF(\tilde{I})$, we deduce
that $\mathbf{\utility}(\tilde{I}, \tilde{\allocationSymbol})$ is
lexicographically greater than $\mathbf{\utility}(I,
\allocationSymbol, \agents \setminus \agents')$.

Let $\allocationSymbol^*: \agents \times \objects \to \mathbb{R}_{\ge
  0}$ be defined by $\allocationSymbol^*(\agent, \object) =
\tilde{\allocationSymbol}(\agent, \object)$ for all $(\agent,
\object)$ in $(\agents \setminus \agents') \times \objects$ and
$\allocationSymbol^*(\agent, \object) = \allocationSymbol(\agent,
\object)$ for all $(\agent, \object)$ in $\agents' \times \objects$.
Thus $\allocationSymbol^*$ belongs to $\allocs{I}$
and $$\mathbf{\utility}(I, \allocationSymbol^*) =
\st{\mathbf{\utility}(I, \allocationSymbol^*, \agents') +
  \mathbf{\utility}(I, \allocationSymbol^*, \agents \setminus
  \agents')} = \st{\mathbf{\utility}(I, \allocationSymbol, \agents') +
  \mathbf{\utility}(\tilde{I}, \tilde{\allocationSymbol})}$$ is
lexicographically greater than $\st{\mathbf{\utility}(I,
  \allocationSymbol, \agents') + \mathbf{\utility}(I,
  \allocationSymbol, \agents \setminus \agents')} =
\mathbf{\utility}(I, \allocationSymbol)$, a contradiction since
$\allocationSymbol$ belongs to $\LMMF(I)$.  
\end{proof}

}

{\bf Game-theoretic desiderata for mechanisms.}
In order to define the strategyproof (SP) and group strategyproof
(GSP) properties, it is convenient to
first define the $k$-SP property for any given positive integer $k$.
An $\oafd$ mechanism is $k$-SP if no coalition of $k$ agents can
misrepresent their demands in such a way that some member of the
coalition gains and no member of the coalition loses.  Formally, an
$\oafd$ mechanism $M$ is said to be $k$-SP if for any $\oafd$ instance
$I=(\agents,\objects,\edw,\supp,\dem)$, any $\allocationSymbol$ in
$M(I)$, any subset $\agents'$ of $\agents$ such that $|\agents'|=k$,
any $\dem^*$ in $\dems{\agents'}{\objects}$, any $\oafd$ instance
$I'=(\agents,\objects,\edw,\supp,\dem')$ where
$\dem'=(\dem_{\agents\setminus\agents'},\dem^*)$, and any
$\allocationSymbol'$ in $M(I')$, either there is no agent $\agent$ in
$\agents'$ such that $\ut{\allocationSymbol}{\dem}{\agent} <
\ut{\allocationSymbol'}{\dem}{\agent}$, or there is an agent $\agent$
in $\agents'$ such that $\ut{\allocationSymbol}{\dem}{\agent} >
\ut{\allocationSymbol'}{\dem}{\agent}$. A mechanism is SP if it is
$1$-SP. A mechanism is GSP if it is $k$-SP for all $k$.

An $\oafd$ mechanism is said to be resource monotonic (RM) if
increasing the supply of one or more objects does not decrease the
utility of any agent.  Formally, an $\oafd$ mechanism $M$ is said to
be RM if for any instances $I = (\agents, \objects, \edw, \supp,
\dem)$ and $I' = (\agents,\objects,\edw,\supp',\dem)$ such that
$\supp(\object)\le \supp'(\object)$ for all objects $\object$ in
$\objects$, any agent $\agent$ in $\agents$, any allocation
$\allocationSymbol$ in $M(I)$, and any allocation $\allocationSymbol'$
in $M(I')$, we have $u(\mu,\dem,a) \le u(\mu',\dem,a)$.

An $\oafd$ mechanism is said to be population monotonic (PM) if 
 decreasing the endowments of one or more agents does not decrease the utility of any other agent.  
Formally, given  any instance $I = (\agents, \objects, \edw, \supp, \dem)$,  we define $\shrink(I)$ as the set of all $\oafd$ instances $I' =(A', B, \edw', \supp,\dem_{ A'})$ such that  
$A'$ is a subset of $A$ and $\edw'(\agent)\le \edw(\agent)$ for all agents $\agent$ in	$A'$.
An $\oafd$ mechanism $M$ is said to be PM if for any $\oafd$ instances  $I = (\agents, \objects, \edw, \supp, \dem)$ and $I' = (A', B, \edw', \supp,\dem_{ A'})$  in $\shrink(I)$, any allocations $\allocationSymbol$ in $M(I)$ and $\allocationSymbol'$ in $M(I')$, and any agent $\agent$ in $ A'$ such that  $\edw'(\agent) = \edw(\agent)$, 
we have $u(\mu,\dem,a) \le u(\mu',\dem,a)$.
		
An $\oafd$ mechanism $M$ is EF (resp., 
NW,  $\CoeffSI$-SI
) if for
any $\oafd$ instance $I$, every allocation in $M(I)$ is EF (resp., 
NW, $\CoeffSI$-SI
).  An $\oafd$ mechanism $M$ is frugal (resp., MMF,
LMMF) if for any $\oafd$ instance $I$, the set of allocations $M(I)$
is contained in $\frugal(I)$ (resp., $\MMF(I)$, $\LMMF(I)$).

\textbf{Lexicographic flow.}
We now briefly review the lexicographic flow problem, which we utilize to obtain an efficient implementation of our mechanism.
We refer readers unfamiliar with the parametric maximum flow to Appendix~\ref{apd:lex_flow} for more details.

Given a flow network $G = (V, \edges)$ with source $s$ and sink $t$, and a subset $S$ of $V - t$ such that $s$ is in $S$, we write $(S, \overline{S})$ to denote the associated cut of $G$.
There is a minimum cut $(S, \overline{S})$ such that $S$ contains $S'$ for all minimum cuts $(S', \overline{S'})$.
We refer to this minimum cut $(S, \overline{S})$ as the source-heavy minimum cut.

In this paper, we consider parametric flow networks  where each edge leaving $s$ has a capacity proportional to a parameter $\lambda$ and all other edge capacities are independent of $\lambda$.
For any parametric flow network $G$, we let $G(\lambda)$ denote the flow network associated with a particular value of $\lambda$.
We use the standard terminology of breakpoints as defined in Gallo et al.~\cite{gallo1989}.
As the value of $\lambda$ increases from $0$, the vertices in $V \setminus \{s,t\}$ move from the sink side to the source side of the source-heavy minimum cut.
For any parametric flow network $G$, 
the breakpoint function $\brkfn{}{v}$ maps 
any given vertex $v$ in $V \setminus \{s,t\}$ to
the breakpoint value of $\lambda$ at which $v$ moves from the sink side to the source side of the source-heavy minimum cut~\cite{stone1978}.

We now define the notion of a lexicographic flow~\cite{megiddo1974,megiddo1977}. 
Assume that the edges leaving $s$ reach the vertices $\{v_1, \dots,
v_k\}$, and that $t$ does not belong to this set.  Let the capacity of
the edge $(s, v_i)$ be $w_i\lambda$.
For a flow $f$ in $G(\infty)$, let $\theta(G, f)$ denote the length-$k$ vector whose $j$th component is the $j$th smallest $f(s,v_i)/w_i$, for $i$ in $[k]$.
A lexicographic flow $f$ of $G$ is a maximum flow $f$ in $G(\infty)$ that is lexicographically at least  $\theta(G,f')$ for all maximum flows $f'$ in $G(\infty)$.
Gallo et al.\ describe an algorithm that computes the breakpoint function and a lexicographic flow in $O(|V||E| \log(|V|^2/|E|))$ time~\cite{gallo1989}.


\section{Frugal Lexicographic Maximin Fair Mechanism}
\label{sec:mechanism}

Let $\LMM$ denote the $\oafd$ mechanism such that $\LMM(I) = \LMMF(I)
\cap \frugal(I)$ for all $\oafd$ instances $I$.  In
Section~\ref{sec:structural} we establish that all agents are
indifferent between allocations in $\LMM(I)$ (see
Lemma~\ref{lem:allocation}).  In this section, we describe an
efficient non-deterministic algorithm $\mathcal{A}$ that implements
$\LMM$ in the following sense: on any input $\oafd$ instance $I$, the
set of possible allocations produced by $\AAA$ is $\LMM(I)$.  The
algorithm $\AAA$ is based on a reduction to the lexicographic flow
problem on a parametric flow network.  Given an $\oafd$ instance $I =
(\agents, \objects, \edw, \supp, \dem)$ as input, algorithm $\AAA$
first creates a parametric flow network $G_I = (\agents \cup \objects
\cup \{s, t\}, \edges)$ with the edge capacities defined by the
functions $\edw, \dem,$ and $\supp_{I}$ described below.  The network
$G_I$ has an agent (resp., object) vertex for each agent (resp.,
object) in the input.  We denote the set of agent (resp., object)
vertices by $\agents$ (resp., $\objects$).  For any agent vertex
$\agent$ and any object vertex $\object$, there is a directed edge of
capacity $\edw(\agent)\lambda$ from $s$ to $\agent$, there is a
directed edge of capacity $\dem(\agent, \object)$ from $\agent$ to
$\object$, and there is a directed edge of capacity
$\suppt{I}{\object}$ from $\object$ to $t$.  It is easy to check that
Observation~\ref{ob:frugal_flow} below holds.
\begin{ob}
\label{ob:frugal_flow}
    For any $\oafd$ instance $I = (\agents, \objects, \edw, \supp, \dem)$, 
	there is a one-to-one correspondence between 
    flows $f$ in $G_I(\infty)$ and allocations $\allocationSymbol$ in $\frugal(I)$ such that 
    $f(\agent, \object) = \allocationSymbol(\agent, \object) $ for all $(\agent, \object)$ in $\agents\times \objects$.
\end{ob}

Given as input an $\oafd$ instance $I$, algorithm $\AAA$
non-deterministically selects a lexicographic flow in $G_I(\infty)$
and outputs the corresponding allocation in $\allocs{I}$.
The algorithm of Gallo et al.\ can be used to compute a lexicographic
flow in $O((|\agents| + |\objects|)|\agents||\objects|
\log((|\agents|+|\objects|)^2 / |\agents||\objects|))$ time.  
Using
Observation~\ref{ob:frugal_flow}, it is straightforward to prove that there is a one-to-one
correspondence between frugal LMMF allocations in $\LMM(I)$ and
lexicographic flows in $G_I$.  Thus we obtain
Lemma~\ref{lem:frugal_LMF} below.
\begin{lemma}\label{lem:frugal_LMF}
	For any $\oafd$ instance $I$, the set of possible allocations produced by algorithm $\AAA$ on input $I$ is equal to $\LMM(I)$.
\end{lemma}
\punt{\begin{proof}
It is straightforward to verify that the following observations hold.
\begin{ob}\label{ob:LMMF_in_NW}
For any $\oafd$ instance $I$, all allocations in $\LMMF(I)$ are NW.
\end{ob}
\begin{ob}\label{ob:min-cut}
For any $\oafd$ instance $I$, the capacity of a minimum cut of
$G_I(\infty)$ is $\supp_I(\objects)$.
\end{ob}
	
We begin by proving the following useful claim.
	
Claim~1: Let $I = (\agents, \objects, \edw, \supp, \dem)$ be an
$\oafd$ instance. Let $\mu$ be an allocation in $\frugal(I)$ and let
$f$ be a flow in $G_I(\infty)$ such that $f(\agent, \object) =
\allocationSymbol(\agent, \object) $ for all $(\agent, \object)$ in
$\agents\times \objects$. Then $f$ is a maximum flow in $G_I(\infty)$ if
and only if $\mu$ is NW.
	
Proof: We first prove the only if direction.  Using the max-flow
min-cut theorem and Observation~\ref{ob:min-cut}, we deduce that the
value of flow $f$ is $\supp_I(\objects)$.  Since $\allocationSymbol$
is frugal and $ \allocationSymbol(\agent, \object)=f(\agent, \object)
$ for all $(\agent, \object)$ in $\agents\times \objects$, we have
$\allocationSymbol(\agents, \objects) = \supp_{I}(\objects)$.  Since
$\allocationSymbol$ is frugal and $\allocationSymbol(\agents,
\objects) = \supp_{I}(\objects)$, we deduce that
$\allocationSymbol(\agents, \object) = \suppt{I}{\object}$ for all
objects $\object$ in $\objects$, which further implies that
$\allocationSymbol$ is NW.  Now, we prove the if direction.  Since
$\mu$ belongs to $\frugal(I)$ and $\allocationSymbol$ is NW, we deduce
that $\allocationSymbol(\agents, \objects) = \supp_{I}(\objects)$.
Thus, the value of flow $f$ is $\supp_{I}(\objects)$.  Hence the
max-flow min-cut theorem and Observation~\ref{ob:min-cut} imply that
flow $f$ is a maximum flow in $G_I(\infty)$.  This concludes the proof
of Claim~1.
	
Let $I = (\agents, \objects, \edw, \supp, \dem)$ be an $\oafd$
instance.  Let $\AAA(I)$ denote the set of possible allocations
produced by algorithm $\AAA$ on input $I$.  Since $\LMM(I) =
\frugal(I)\cap \LMMF(I)$, it suffices to prove that $\AAA(I) =
\frugal(I)\cap \LMMF(I)$.  Thus Claims~2 and~3 below imply that the
lemma holds.
	
Claim~2: $\frugal(I)\cap \LMMF(I)\subseteq \AAA(I)$.

Proof: Let $\mu$ be an allocation in $\frugal(I)\cap \LMMF(I)$.
Observation~\ref{ob:frugal_flow} implies that there is a flow in
$G_I(\infty)$, call it $f$, such that $\allocationSymbol(\agent,
\object) = f(\agent, \object)$ for all $(a,\object)$ in $\agents\times \objects$.
Observation~\ref{ob:LMMF_in_NW} implies that $\mu$ is NW, and hence
Claim~1 implies that flow $f$ is a maximum flow in $G_I(\infty)$.  To
prove that $\mu$ belongs to $\AAA(I)$, it suffices to prove that
$f$ is a lexicographic flow in $G_I$.  Assume for the sake of
contradiction that $f$ is not a lexicographic flow in $G_I$.  Hence
there is a maximum flow $f'$ in $G_I(\infty)$ such that
$\theta(G_I,f')$ is lexicographically greater than $\theta(G_I,f)$.
Observation~\ref{ob:frugal_flow} implies that there is a frugal
allocation, call it $\allocationSymbol'$, such that
$\allocationSymbol'(\agent, \object) = f(\agent, \object)$ for all
$(\agent, \object)$ in $\agents\times \objects$.  Since
$\allocationSymbol(\agent, \object) = f(\agent, \object)$ and
$\allocationSymbol'(\agent, \object) = f'(\agent, \object)$ for all
$(a,\object)$ in $\agents\times \objects$, we deduce that
$\mathbf{\utility}(I,\allocationSymbol) =\theta(G_I,f)$ and
$\mathbf{\utility}(I,\allocationSymbol') =\theta(G_I,f')$,
respectively.  Since $\mathbf{\utility}(I,\allocationSymbol)
=\theta(G_I,f)$, $\mathbf{\utility}(I,\allocationSymbol')
=\theta(G_I,f')$, and $\theta(G_I,f')$ is lexicographically greater
than $\theta(G_I,f)$, we deduce that
$\mathbf{\utility}(I,\allocationSymbol')$ is lexicographically greater
than $\mathbf{\utility}(I,\allocationSymbol)$, a contradiction since
$\mu$ belongs to $ \LMMF(I)$. This concludes the proof of Claim~2.

Claim~3: $\AAA(I) \subseteq \frugal(I)\cap \LMMF(I)$.

Proof: Let $\allocationSymbol$ be an allocation function in $\AAA(I)$.
Let $f$ denote the lexicographic flow in $G_I$ selected by algorithm
$\AAA$; thus $f$ corresponds to $\allocationSymbol$.  Since any
lexicographic flow in $G_I$ is a maximum flow in $G_I(\infty)$, we
deduce from Observation~\ref{ob:frugal_flow} that $\allocationSymbol$
belongs to $\frugal(I)$.  It remains to prove that $\mu$ belongs to
$\LMMF(I)$. Assume for the sake of contradiction that
$\allocationSymbol$ does not belong to $\LMMF(I)$.  Hence there is an
allocation $\allocationSymbol'$ in $\LMMF(I)$ such that
$\mathbf{\utility}(I,\allocationSymbol')$ is lexicographically greater
than $\mathbf{\utility}(I,\allocationSymbol)$.  Let
$\allocationSymbol''$ be an allocation in $\allocs{I}$ such that
$\mu''(a,\object) = \min(\mu'(a,\object),d(a,\object))$ for all $(a,\object)$ in $\agents\times \objects$.  The
definition of $\allocationSymbol''$ implies that $\allocationSymbol''$
belongs to $\frugal(I)$.  Since the maximum utility an agent $a$ can
achieve from an object $b$ is $d(a, \object)$, we have
$\mathbf{\utility}(I,\allocationSymbol'') =
\mathbf{\utility}(I,\allocationSymbol')$.  Since
$\mathbf{\utility}(I,\allocationSymbol'') =
\mathbf{\utility}(I,\allocationSymbol')$ and $\allocationSymbol'$
belongs to $\LMMF(I)$, we deduce that $\allocationSymbol''$ belongs to
$\LMMF(I)$.  Since $\allocationSymbol''$ belongs to $\LMMF(I)$,
Observation~\ref{ob:LMMF_in_NW} implies that $\allocationSymbol''$ is
NW.  Since $\mathbf{\utility}(I,\allocationSymbol')$ is
lexicographically greater than
$\mathbf{\utility}(I,\allocationSymbol)$ and
$\mathbf{\utility}(I,\allocationSymbol'') =
\mathbf{\utility}(I,\allocationSymbol')$, we conclude that
$\mathbf{\utility}(I,\allocationSymbol'')$ is lexicographically
greater than $\mathbf{\utility}(I,\allocationSymbol)$.  Since
$\allocationSymbol''$ belongs to $\frugal(I)$ and
$\allocationSymbol''$ is NW, Observation~\ref{ob:frugal_flow} and
Claim~1 imply that there is a maximum flow in $G_I(\infty)$, call it
$f''$, such that $f''(\agent, \object) = \allocationSymbol''(\agent,
\object)$ for all $(\agent, \object)$ in $\agents\times\objects$.  Since
$\allocationSymbol(\agent, \object) = f(\agent, \object)$ and
$\allocationSymbol''(\agent, \object) = f''(\agent, \object)$ for all
$(a,\object)$ in $\agents\times \objects$, we deduce that
$\mathbf{\utility}(I,\allocationSymbol) =\theta(G_I,f)$ and
$\mathbf{\utility}(I,\allocationSymbol'') =\theta(G_I,f'')$,
respectively.  Since $\mathbf{\utility}(I,\allocationSymbol) =
\theta(G_I,f)$, $\mathbf{\utility}(I,\allocationSymbol'') =
\theta(G_I,f'')$, and $\mathbf{\utility}(I,\allocationSymbol'')$ is
lexicographically greater than
$\mathbf{\utility}(I,\allocationSymbol)$, we deduce that
$\theta(G_I,f'')$ is lexicographically greater than $\theta(G_I,f)$, a
contradiction since $f$ is a lexicographic flow in $G_I$.  This
concludes the proof of Claim~3.  
\end{proof}
}

We introduce some notations that are helpful in analysis of mechanism
$\LMM$.  The algorithm of Gallo et al.\ to find the lexicographic flow
computes the breakpoint function $\Lambda_I$ such that
$\brkfn{I}{\agent}$ denotes the breakpoint at which agent vertex
$\agent$ moves from the sink side to the source side of the
source-heavy minimum cut, for all agent vertices $\agent$ in
$\agents$.  Let $\num{I}$ denote $|\{\brkfn{I}{\agent} \mid \agent \in
\agents\}|$.  For any $i$ in $[\num{I}]$, let $\brkpts{I}{i}$ denote
the $i$th smallest value in $\{\brkfn{I}{\agent} \mid \agent \in
\agents\}$, and let $\agnts{I}{i}$ denote the set $\{\agent \in
\agents \mid \brkfn{I}{\agent} \le \brkpts{I}{i}\}$.  We define
$\agnts{I}{0}$ as $\emptyset$.  For any $i$ in $[\num{I}]$, and object
$\object$ in $\objects$, let $\caap(I,i,\object)$ denote
$\suppt{I}{\object} - \dem(\agnts{I}{i-1}, \object)$.  With
$\objs{I}{0}$ defined as $\emptyset$, for any $i$ in $[\num{I}]$, let
$\objs{I}{i}$ be recursively defined as the union of $\objs{I}{i-1}$
and
\begin{equation*}
\{\object \in \objects \setminus \objs{I}{i-1} \mid \dem(\agnts{I}{i} \setminus \agnts{I}{i-1}, \object) > \caap(I,i,\object)\}.
\end{equation*}

\subsection{Technical Properties of Mechanism $\LMM$}
\label{sec:structural}

In this section, we establish some basic technical results concerning
mechanism $\LMM$.  These results are used in
Section~\ref{sec:allocProps} (resp., Section~\ref{sec:GSP}) to derive
certain game-theoretic properties of frugal LMMF allocations (resp.,
mechanisms).  Throughout this section, let $I = (\agents, \objects,
\edw, \supp, \dem)$ denote an $\oafd$ instance and let $G$ denote
$G_I$.  We let $\Lambda$ and $k$ denote the breakpoint function
$\Lambda_I$ and the value $\num{I}$, respectively.  For any $i$ in
$[k]$, we let $\brkpt_i$, $\agents_i$, and $\objects_i$ denote
$\brkpts{I}{i}$, $\agnts{I}{i}$, and $\objs{I}{i}$, respectively.  In
addition, let $\agents_0$ denote $\agnts{I}{0}$ and let $\objects_0$
denote $\objs{I}{0}$.  For any $i$ in $[k]$ and any object $\object$
in $\objects$, we let $\capt_i(\object)$ denote $\caap(I, i, b)$.  For
any $i$ in $[k]$, and any non-empty subset $A'$ of $\agents \setminus
\agents_{i-1}$, let $\capacity_i(A')$ denote $\sum_{\object \in
  \objects \setminus \objects_{i-1}} \min \left(\capt_i(\object),
\dem(A', \object)\right)$.  Notice that for any non-empty subset $A'$
of $\agents$, we have $\capacity_1(\agents') = \capty{I}{\agents'}$.

\punt{Before proving Lemma~\ref{lem:lambda}, we show Lemma~\ref{lem:aux}
below which establishes important properties of the minimum
breakpoint of any agent vertex in each parametric flow network $G_i$.
Throughout this section, we make the following definitions for all $i$ in $[k]$:
$\Lambda^*_i$ denotes the breakpoint function of $G_i$; 
$\brkpt^*_i$ denotes the minimum breakpoint of
an agent vertex in $G_i$;
$\brkpt^{**}_i$ denotes $\min_{\emptyset \neq A' \subseteq \agents \setminus \agents_{i-1}} \capacity_i(A')/\edw(A')$;
$\agents^*_i$ denotes $\bigcup \{A' \subseteq \agents \setminus \agents_{i-1} \mid \capacity_i(A')
= \edw(A')\brkpt^*_i\}$; 
$\objects^*_i$ denotes $\{\object \in \objects \setminus \objects_{i-1} \mid \dem(\agents^*_i, \object)
> \capt_i(\object)\}$;
 $\Psi_1(i)$ denotes the predicate
``$\brkpt^*_i=\brkpt^{**}_i$'';
$\Psi_2(i)$ denotes the predicate ``$\Lambda^*_i(\agent) = \brkpt^*_i$ for all agents $\agent$ in $\agents^*_i$''; 
$\Psi_3(i)$ denotes the predicate ``for any flow in
$G_i(\infty)$ such that $f(s, \agent)=\edw(\agent)\brkpt^*_i$ for all
agent vertices $\agent$ in $\agents \setminus \agents_{i-1}$, we have
$f(\agent, \object) = \dem(\agent, \object)$ for all
$(\agent, \object)$ in $\agents^*_i \times
((\objects \setminus \objects_{i-1}) \setminus \objects^*_i)$'';
$\Psi_4(i)$ denotes the predicate ``for any flow in $G_i(\infty)$ such that
$f(s, \agent)=\edw(\agent)\brkpt^*_i$ for all agent vertices $\agent$ in
$\agents \setminus \agents_{i-1}$, we have $f(\agents^*_i, \object) =
f(\object, t) = \capt_i(\object)$ for all object vertices $\object$ in
$\objects^*_i$.''

\begin{lemma}
\label{lem:aux}
Let $i$ be in $[k]$.
Then predicate $\Psi_j(i)$ holds for all $j$ in
$\{1,\ldots,4\}$.
\end{lemma}
\begin{proof}
Recall that the set of agent (resp., object) vertices in $G_i$ is
$\agents \setminus \agents_{i-1}$ (resp.,
$\objects \setminus \objects_{i-1})$.  
We first establish the
following useful claim, which implies that $\brkpt^*_i \ge \brkpt^{**}_i$.
	
Claim~1: A maximum flow in $G_i(\brkpt^{**}_i)$ has value
$\edw(\agents \setminus \agents_{i-1})\brkpt^{**}_i$.
	
Proof: To prove that a maximum flow in $G_i(\brkpt^{**}_i)$ has value
$\edw(\agents \setminus \agents_{i-1})\brkpt^{**}_i$, it is sufficient to
argue that a minimum cut in $G_i(\brkpt^{**}_i)$ has capacity
$\edw(\agents \setminus \agents_{i-1})\brkpt^{**}_i$.  Let the source-heavy
minimum cut in $G_i(\brkpt^{**}_i)$ be $(S, \overline{S})$ and let
$\agents'$ denote $S \cap (\agents \setminus \agents_{i-1})$.  We
begin by showing that $S \cap (\objects \setminus \objects_{i-1})
= \{\object \in \objects \setminus \objects_{i-1} \mid \dem(\agents', \object) \ge \capt_i(\object)\}$.
Assume for the sake of contradiction that 
this equation does not hold.
We consider three cases.
	
Case~1: There is an object vertex $\object$ in
$\overline{S}\cap(\objects \setminus \objects_{i-1})$ such that
$\dem(\agents', \object) > \capt_i(\object)$.  Hence the capacity
of the cut $(S+\object, \overline{S}-\object)$ is
$\dem(\agents', \object) - \capt_i(\object) > 0$ less than the
capacity of the cut $(S, \overline{S})$, a contradiction since
$(S, \overline{S})$ is a minimum capacity cut.
	
Case~2: There is an object vertex $\object$ in
$S\cap(\objects \setminus \objects_{i-1})$ such that $\capt_i(\object)
> \dem(\agents', \object)$.  Hence the capacity of the cut
$(S-\object, \overline{S}+\object)$ is $\capt_i(\object)
- \dem(\agents', \object) > 0$ less than the capacity of the cut
$(S, \overline{S})$, a contradiction since $(S, \overline{S})$ is a
minimum capacity cut.
	
Case~3: There is an object $\object$ in
$\overline{S}\cap(\objects \setminus \objects_{i-1})$ such that
$\capt_i(\object) = \dem(\agents', \object)$.  Hence the cuts
$(S+\object, \overline{S}-\object)$ and $(S, \overline{S})$ have the
same capacity, but $(S+\object, \overline{S}-\object)$ has a larger
source side, a contradiction.
	
From the above case analysis, $S \cap
(\objects \setminus \objects_{i-1})
= \{\object \in \objects \setminus \objects_{i-1} \mid \dem(\agents', \object) \ge \capt_i(\object)\}$.
Thus the capacity of the cut $(S, \overline{S})$ is 
\begin{equation*}
\edw(\agents \setminus (\agents_{i-1}\cup\agents'))\brkpt^{**}_i + \sum_{\object \in \objects \setminus \objects_{i-1}} \min(\capt_i(\object), \dem(\agents', \object))
		= \edw(\agents \setminus (\agents_{i-1}\cup\agents'))\brkpt^{**}_i + \capacity_i(\agents').
\end{equation*}
The definition of $\brkpt^{**}_i$ implies that
$\edw(\agents') \brkpt^{**}_i \le \capacity_i(\agents')$.  Thus the
capacity of the minimum cut is at least
$\edw(\agents \setminus \agents_{i-1})\brkpt^{**}_i$.  Moreover, the capacity
of cut $({s}, V\setminus{s})$ is
$\edw(\agents \setminus \agents_{i-1})\brkpt^{**}_i$. Thus the capacity of a
minimum cut of $G_i(\brkpt^{**}_i)$ is
$\edw(\agents \setminus \agents_{i-1})\brkpt^{**}_i$.  This concludes the
proof of Claim~1.
	
Let $\brkpt'$ be a value greater than $\brkpt^{**}_i$.  We show that there is
no flow in $G_i(\brkpt')$ such that every agent vertex $\agent$ has
incoming flow $\edw(\agent)\brkpt'$.  Assume for the sake of
contradiction that there is a flow such that every agent vertex
$\agent$ has incoming flow $\edw(\agent)\brkpt'$.
The total capacity of the edges leaving $\agents^*_i \cup \objects^*_i$ is
\begin{equation*}
	\sum_{\object \in \objects \setminus \objects_{i-1}} \min(\capt_i(\object), \dem(\agents^*_i, \object))
	= \capacity_i(\agents^*_i). 
\end{equation*}
Since $\capacity_i(\agents^*_i)
= \edw(\agents^*_i)\brkpt^{**}_i < \edw(\agents^*_i)\brkpt'$,
the total capacity of
the edges leaving $\agents^*_i \cup \objects^*_i$ is less than the total
flow into the set $\agents^*_i \cup \objects^*_i$, a contradiction.  This
result, together with Claim~1, establishes that $\Psi_1(i)$ and
$\Psi_2(i)$ hold.

Let $f$ be a flow in $G_i(\infty)$ such that
$f(s, \agent)=\edw(\agent)\brkpt^*_i$ for all agent vertices $\agent$ in
$\agents \setminus \agents_{i-1}$
Since the total capacity of the edges leaving
$\agents^*_i \cup \objects^*_i$ is $\capacity_i(\agents^*_i)$, which is equal
to the total flow $\edw(\agents^*_i)\brkpt^*_i$ into
$\agents^*_i \cup \objects^*_i$ in $f$, we deduce that $f(\edge)
= \capt(\edge)$ for all edges $\edge$ leaving
$\agents^*_i \cup \objects^*_i$. Thus $f(\object, t) = \capt_i(\object)$ for
all object vertices $\object$ in $\objects^*_i$, and $f(\agent, \object)
= \dem(\agent, \object)$ for all $(\agent, \object)$ in
$\agents^*_i \times
((\objects \setminus \objects_{i-1}) \setminus \objects^*_i)$.
Moreover, since the
total flow into $\agents^*_i \cup \objects^*_i$ is
$\capacity_i(\agents^*_i)= \capt_i(\objects^*_i) + \dem(\agents^*_i,
(\objects \setminus \objects_{i-1}) \setminus \objects^*_i)$, we have
$f(\agents^*_i, \objects^*_i) = \capt_i(\objects^*_i)$.  
It follows that $f(\agents^*_i, \object) = f(\object, t)$ for all object vertices $\object$ in $\objects^*_i$.  We conclude that $\Psi_3(i)$
and $\Psi_4(i)$ hold.
\end{proof}
}

Lemma~\ref{lem:lambda} characterizes the values of the breakpoints of
$G$, and the breakpoint associated with each agent vertex.  It also
establishes a connection between a lexicographic flow in $G$ and the
sets $\agents_1, \dots, \agents_k, \objects_1, \dots, \objects_k$.
The result of Lemma~\ref{lem:lambda} is similar in spirit to
Theorem~4.6 of Megiddo \cite{megiddo1974} for general parametric flow
networks.  Since we work with parametric flow networks with a special
structure, we are able to obtain a more specific result and we can
characterize a lexicographic flow in greater detail.  Our proof of
Lemma~\ref{lem:lambda} is not based on Megiddo's proof; instead, we
provide a simpler proof for our special case.  Our formulation of
Lemma~\ref{lem:lambda} generalizes Megiddo's result in one aspect,
since it allows for agents with different endowments; this
generalization is straightforward.

Consider a sequence of parametric flow networks $G_1, \dots, G_{k}$,
where $G_i$ is the subgraph of $G$ induced by $(\agents \setminus
\agents_{i-1}) \cup (\objects \setminus \objects_{i-1}) \cup \{s,t\}$,
except that for any object vertex $\object$ in $G_i$, the capacity of
edge $(\object, t)$ is defined to be $\capt_i(\object)$.  Remark: It
follows easily from Lemma~\ref{lem:lambda} below that
$\capt_i(\object)\geq 0$.

For any $i$ in $[k]$, we define the following predicates:
$\Gamma_1(i)$ denotes ``the minimum breakpoint of any agent vertex in
$G_i$ is $\brkpt_i$''; $\Gamma_2(i)$ denotes ``$\brkpt_i$ is equal to
$\min_{\emptyset\neq \tilde{A} \subseteq\agents\setminus\agents_{i-1}}
\capacity_i(\tilde{A})/\edw(\tilde{A})$''; $\Gamma_3(i)$ denotes
``$\agents_i$ is equal to $\agents_{i-1} \cup \bigcup \{\tilde{A}
\subseteq \agents \setminus \agents_{i-1} \mid \capacity_i(\tilde{A})
= \edw(\tilde{A})\brkpt_i\}$''; $\Gamma_4(i)$ denotes ``for any
lexicographic flow $f$ in $G$, we have $f(\agent, \object) =
\dem(\agent, \object)$ for all $(\agent, \object)$ in $(\agents_i
\setminus \agents_{i-1}) \times (\objects \setminus \objects_i)$'';
$\Gamma_5(i)$ denotes ``for any lexicographic flow $f$ in $G$, we have
$f(\agents_i,\object)=f(\object,t)=\suppt{I}{\object}$ and $f(\agent,
\object)=0$ for all $(\agent, \object)$ in
$(\agents\setminus\agents_i)\times(\objects_i\setminus\objects_{i-1})$.''
\begin{lemma}\label{lem:lambda}
  Predicate $\Gamma_j(i)$ holds for all $i$ in $[k]$ and all $j$ in $\{1,\ldots,5\}$.
\end{lemma}
\punt{\begin{proof}
	Let $f$ denote a lexicographic flow in $G$ and for any $i$ in $[k]$, let $P(i)$ denote the predicate ``$\Gamma_j(i)$ holds for all $j$ in $\{1, \ldots, 5\}$.''
	We prove by induction that $P(i)$ holds for all $i$ in $[k]$.

	Base case:
	Since $\agents_0 = \emptyset$, we have $\capt_1(\object) = \suppt{I}{\object}$ for all object vertices $\object$ in $\objects$. 
	Thus $G_1 = G$, and hence $\Gamma_1(1)$ holds.
	Lemma~\ref{lem:aux} implies that $\Psi_1(1)$ and $\Psi_2(1)$ hold; hence $\Gamma_2(1)$ and $\Gamma_3(1)$ hold.
	Since $\Psi_3(1)$ and $\Psi_4(1)$ hold by Lemma~\ref{lem:aux}, $\brkpt^*_1 = \brkpt_1$, $\agents^*_1 = \agents_1 \setminus \agents_0$, and $\objects^*_1 = \objects_1 \setminus \objects_0$, we deduce that $\Gamma_4(1)$ and $\Gamma_5(1)$ hold.

	Induction step: Let $i$ belong to $\{2, \dots, k\}$ and 
	assume that $P(i')$ holds for all $i'$ in $[i-1]$.
	We need to prove that $P(i)$ holds. 
	Let $\object$ be an object vertex in $\objects \setminus \objects_{i-1}$.
	Since the induction hypothesis implies that $\Gamma_4(i')$ holds for all  $i'$ in $[i-1]$, we deduce that $f(\agent, \object) = \dem(\agent, \object)$ for all $\agent$ in $\agents_{i-1}$.
	Thus $f(\agents \setminus \agents_{i-1}, \object) \le \suppt{I}{\object} - \dem(\agents_{i-1}, \object) = \capt_i(\object)$.
	Moreover, since the induction hypothesis implies that $\Gamma_5(i')$ holds for all $i'$ in $[i-1]$, we deduce that $f(\agent, \object') = 0$ for all $(\agent, \object')$ in $(\agents \setminus \agents_{i-1}) \times \objects_{i-1}$.
	From the aforementioned results, it is straightforward to verify that $\Gamma_1(i)$ holds.
	Lemma~\ref{lem:aux} implies that $\Psi_1(i)$ and $\Psi_2(i)$ hold; hence $\Gamma_2(i)$ and $\Gamma_3(i)$ hold.
	Since $\Psi_3(i)$ holds by Lemma~\ref{lem:aux}, $\brkpt^*_i = \brkpt_i$, $\agents^*_i = \agents_i \setminus \agents_{i-1}$, and $\objects^*_i = \objects_i \setminus \objects_{i-1}$, we deduce that
	$\Gamma_4(i)$ holds.
	Let $\object'$ be an object vertex in $\objects_i \setminus \objects_{i-1}$.
	Predicate $\Psi_4(i)$ implies that $f(\agents_i \setminus \agents_{i-1}, \object') = c_i(\object') = \suppt{I}{\object'} - \dem(\agents_{i-1}, \object')$.
	Since the induction hypothesis implies $\Gamma_4(i')$ holds for all $i'$ in $[i-1]$, we deduce that $f(\agent, \object') = \dem(\agent, \object')$ for all $\agent$ in $\agents_{i-1}$.
	Hence $f(\agents_{i-1}, \object') = \dem(\agents_{i-1}, \object')$.
	Since $f(\agents_{i-1}, \object') = \dem(\agents_{i-1}, \object')$ and $f(\agents_i \setminus \agents_{i-1}, \object') = \suppt{I}{\object'} - \dem(\agents_{i-1}, \object')$, we deduce that $f(\agents_i, \object') = \suppt{I}{\object'} = f(\object', t)$, where the last equality holds because the capacity of edge $(\object', t)$ is $\suppt{I}{\object'}$.
	Since $f(\agents_i, \object') = \suppt{I}{\object'}$ and $\suppt{I}{\object'}$ is the capacity of edge $(\object', t)$, we deduce that $f(\agent, \object) = 0$ for all $\agent$ in $\agents \setminus \agents_i$, which establishes $\Gamma_5(i)$.
	We conclude that $P(i)$ holds, as required.
\end{proof}

}


We now prove some results about frugal LMMF allocations.  Throughout
the remainder of the section, let $\allocationSymbol$ denote an
allocation in $\LMM(I)$.  The definition of mechanism $\LMM$ implies
that $\allocationSymbol$ is frugal and LMMF.  Recall that algorithm
$\AAA$ first computes a lexicographic flow $f$ in $G$ such that
$\allocationSymbol(\agent, \object) = f(\agent, \object)$ for all
$(\agent,\object)$ in $\agents\times\objects$.

Corollaries~\ref{cor:under_demands_rounds}
and~\ref{cor:over_demands_rounds} below describe structural properties
of the allocation $\allocationSymbol$ that follow immediately from
predicates~$\Gamma_4$ and~$\Gamma_5$, respectively.
\begin{corollary}\label{cor:under_demands_rounds}
	For any $i$ in $[k]$, any agent $\agent$ in $\agents_i \setminus \agents_{i-1}$, and any object $\object$ in $\objects \setminus \objects_i$, we have $\allocationSymbol(\agent,\object) = \dem(\agent,\object)$. 
\end{corollary}
\begin{corollary}\label{cor:over_demands_rounds}
	For any $i$ in $[k]$, any agent $\agent$ in $\agents \setminus \agents_{i}$, and any object $\object$ in $\objects_{i}$, we have
	$\allocationSymbol(\agent,\object)=0$.
\end{corollary}

Lemma~\ref{lem:allocation} below establishes some basic results that
are useful for many of our subsequent proofs.  For example, we use
Lemma~\ref{lem:allocation} along with
Corollaries~\ref{cor:under_demands_rounds}
and~\ref{cor:over_demands_rounds} to prove that any frugal LMMF
allocation is EF (Theorem~\ref{thm:EF}).
\begin{lemma}
\label{lem:allocation}
Let $\agent$ be an agent in $\agents$ and let $\object$ be an object
in $\objects$.  Then, $\allocationSymbol(\agent, \object)$ belongs to
$[0,\dem(\agent,\object)]$, and $\ut{\allocationSymbol}{\dem}{\agent}
= \allocationSymbol(\agent,\objects) = \edw(\agent)\brkfn{}{\agent}$.
\end{lemma}
\punt{\begin{proof}
	The capacity of edge $(\agent, \object)$ in $G$ is $\dem(\agent, \object)$. 
	Hence $\allocationSymbol(\agent, \object) = f(\agent, \object)$ belongs to $[0,\dem(\agent,\object)]$.
	Flow $f$ satisfies $\edw(\agent)\brkfn{}{\agent} = f(s, \agent) = f(\agent, \objects) = \allocationSymbol(\agent, \objects) = \ut{\allocationSymbol}{\dem}{\agent}$.
\end{proof}
}

We use Lemma~\ref{lem:leeq} below, along with
Lemma~\ref{lem:allocation} and
Corollary~\ref{cor:under_demands_rounds}, to prove that any frugal
LMMF allocation is $\frac12$-SI (Theorem~\ref{thm:SI}).
\begin{lemma}\label{lem:leeq}
	Let $i$ be in $[k]$.
	Then
	$\sum_{j\in [i]} \edw(\agents_j \setminus \agents_{j-1}) \lambda_j = \supp(\objects_i)+ \dem(\agents_i,\objects\setminus \objects_i)$. 
\end{lemma}
\punt{\begin{proof}
	We begin by proving the following useful claim.

	Claim~1: Let $i$ belong to $[k]$. 
	Let $\object$ be an object in $\objects_i \setminus \objects_{i-1}$. 
	Then $\suppt{I}{\object} = \supp(\object)$.
	
	To prove Claim~1, observe that the definition of $\objects_i \setminus \objects_{i-1}$ implies that $\dem(\agents_i \setminus \agents_{i-1}, \object) > \capt_i(\object) = \suppt{I}{\object} - \dem(\agents_{i-1}, \object)$.
	Thus $\dem(\agents, \object) \ge \dem(\agents_i, \object) > \suppt{I}{\object}$.
	Since $\dem(\agents, \object) > \suppt{I}{\object}$ and $\suppt{I}{\object} = \min(\supp(\object), \dem(\agents, \object))$, we have $\suppt{I}{\object} = \supp(\object)$.
	This completes the proof of Claim~1.

	Notice that $\sum_{j\in [i]} \edw(\agents_j \setminus \agents_{j-1}) \lambda_j$ is the total flow into $\agents_i$ in $f$.
	The total flow out of $\agents_i$ in $f$ is $f(\agents_i, \objects)$.
	For any agent vertex $\agent$ in $\agents_i$ and any object vertex $\object$ in $\objects \setminus \objects_{i}$, Lemma~\ref{lem:lambda} implies that $f(\agent, \object) = \dem(\agent, \object)$.
	For any object $\object$ in $\objects_i$, Lemma~\ref{lem:lambda} implies that $f(\agents_i, \object) = \suppt{I}{\object}$. 
	Thus $f(\agents_i, \objects) = \suppt{I}{\objects_i} + \dem(\agents_i,\objects\setminus \objects_i)$.
	Since the net flow into $\agents_i$ is $0$, we obtain $\sum_{j\in [i]} \edw(\agents_j \setminus \agents_{j-1}) \lambda_j = \suppt{I}{\objects_i} + \dem(\agents_i,\objects\setminus \objects_i) = \supp(\objects_i) + \dem(\agents_i,\objects\setminus \objects_i)$, where the last equality follows from Claim~1.
\end{proof}}

We use Lemma~\ref{lem:frozenAgents_at_diff_iter}
below, along with
Lemma~\ref{lem:allocation} and the result that any frugal LMMF
allocation is NW (Theorem~\ref{thm:NW}), to prove that any frugal LMMF
mechanism is RM (Theorem~\ref{thm:RM}).
We use Lemmas~\ref{lem:suboptimal},~\ref{lem:lambda},
and~\ref{lem:allocation} and the result that any frugal LMMF
mechanism is RM (Theorem~\ref{thm:RM}) to prove that any frugal LMMF
mechanism is PM (Theorem~\ref{thm:PM}).
\begin{lemma}\label{lem:frozenAgents_at_diff_iter}
	Let $i$ be in $[k]$, let $a$ (resp., $a'$) be an agent in
	$\agents_i$ (resp., $\agents \setminus \agents_i$), and let
	$\object$ be an object in $\objects$ such that
	$\allocationSymbol(\agent',\object)>0$.  Then
	$\allocationSymbol(\agent,\object)=\dem(\agent,\object)$.
\end{lemma}
\punt{\begin{proof}
	Corollary~\ref{cor:over_demands_rounds} and $\allocationSymbol(\agent',\object)>0$
	imply that $\object$ belongs to $\objects \setminus \objects_i$.
	Hence Corollary~\ref{cor:under_demands_rounds}
	implies that $\allocationSymbol(\agent,\object)=\dem(\agent,\object)$.
\end{proof}}

We use Lemma~\ref{lem:minimal_allocation} below, along with
Lemmas~\ref{lem:suboptimal},~\ref{lem:lambda},
and~\ref{lem:allocation} and the result that any frugal LMMF
mechanism is RM (Theorem~\ref{thm:RM}), to prove that any frugal LMMF
mechanism is GSP (Theorem~\ref{thm:GSP}).
\begin{lemma}\label{lem:minimal_allocation}
Let $i$ belong to $[k]$ and let $\allocationSymbol'$ be an allocation
in $\allocs{I}$ such that
$\ut{\allocationSymbol'}{\dem}{a}\ge \edw(a)\brkfn{}{a}$ for all
agents $\agent$ in $\agents_i$.  Then
$\allocationSymbol'(\agents_i, \object) \ge \allocationSymbol(\agents_i, \object)$
for all objects $\object$ in $\objects$.
\end{lemma}
\punt{\begin{proof}
	For any object $\object$ in $\objects$, we let $\utility(\object)$ (resp., $\utility'(\object)$) denote $\sum_{\agent \in \agents_i} \utobj{\allocationSymbol}{\dem}{\agent}{\object}$ (resp., $\sum_{\agent \in \agents_i} \utobj{\allocationSymbol'}{\dem}{\agent}{\object}$).
	We begin by establishing a useful claim.

	Claim~1: We have $\utility'(\object) \le \utility(\object)$ for all $\object$ in $\objects$.

	Proof:
	Let $\object$ be an object in $\objects$.
	We consider two cases.

	Case~1: $\object \in \objects_i$.
	The definition of $\suppt{I}{\object}$ implies that $\utility'(\object) \le \suppt{I}{\object}$.
	Using Lemma~\ref{lem:lambda} and the definition of $\allocationSymbol$, we deduce that $\allocationSymbol(\agents_i, \object) = f(\agents_i, \object) = \suppt{I}{\object}$.
	Lemma~\ref{lem:allocation} implies that $\allocationSymbol(\agent, \object)$ belongs to $[0, \dem(\agent, \object)]$ for all $\agent$ in $\agents_i$.
	Using Lemma~\ref{lem:allocation}, we conclude that $\utility(\object) = \allocationSymbol(\agents_i, \object) = \suppt{I}{\object}$.
	Thus $\utility'(\object) \le \utility(\object)$.

	Case~2: $\object \in \objects \setminus \objects_i$.
	We have $\utility'(\object) \le \dem(\agents_i, \object)$.
	Using Lemma~\ref{lem:lambda} and the definition of $\allocationSymbol$, we deduce that for any agent $\agent$ in $\agents_i$, $\allocationSymbol(\agent, \object) = f(\agent, \object) = \dem(\agent, \object)$.
	Using Lemma~\ref{lem:allocation}, we conclude that $\utility(\object) = \dem(\agents_i, \object)$.
	Thus, $\utility'(\object) \le \utility(\object)$.

	We have 
	\[\sum_{\object \in \objects} \utility'(\object) = \sum_{\agent \in \agents_i} \ut{\allocationSymbol'}{\dem}{\agent} \ge \sum_{\agent \in \agents_i} \brkfn{}{\agent}\edw(\agent) =  \sum_{\agent \in \agents_i} \ut{\allocationSymbol}{\dem}{\agent} = \sum_{\object \in \objects} \utility(\object),\] 
	where the second equality follows from Lemma~\ref{lem:allocation}.
	Together with Claim~1,
	we deduce that $\utility'(\object) = \utility(\object)$ for all $\object$ in $\objects$.
	Thus $\allocationSymbol'(\agents_i, \object) \ge \utility'(\object) = \utility(\object) = \allocationSymbol(\agents_i, \object)$ for all $\object$ in $\objects$, where the last equality follows from Lemma~\ref{lem:allocation}.
\end{proof}
}

\subsection{Game-Theoretic Properties of Frugal LMMF Allocations}
\label{sec:allocProps}

In this section we establish some game-theoretic properties of frugal
LMMF allocations.  Throughout this section, let $I = (\agents,
\objects, \edw, \supp, \dem)$ be an $\oafd$ instance and let
$\allocationSymbol$ belong to $\LMM(I)$.  The definition of $\LMM$
implies that $\allocationSymbol$ is an arbitrary frugal LMMF
allocation. 

\begin{theorem}
\label{thm:NW}
Allocation $\allocationSymbol$ is NW.
\end{theorem}
\punt{\begin{proof}
	The definition of $\allocationSymbol$ implies that $\allocationSymbol$ belongs to $\frugal(I)$.
	Assume for the sake of contradiction that $\allocationSymbol$ is not NW.
	Hence there is an agent $\agent$ in $\agents$ and an object $\object$ in $\objects$ such that $\allocationSymbol(\agent, \object) < \dem(\agent, \object)$ and $\allocationSymbol(\agents, \object) < \supp(\object)$.
	Let $\allocationSymbol'$ be the allocation in $\allocs{I}$ such that $\allocationSymbol'(\agent, \object) = \min(\dem(\agent, \object), \supp(\object) - \allocationSymbol(\agents - \agent, \object))> \allocationSymbol(\agent, \object)$, and $\allocationSymbol'(\agent', \object') = \allocationSymbol(\agent', \object')$ for all $(\agent', \object')$ in $\agents \times \objects - (\agent, \object)$.
	Thus $\ut{\allocationSymbol'}{\dem}{\agent} > \ut{\allocationSymbol}{\dem}{\agent}$ and $\ut{\allocationSymbol'}{\dem}{\agent'} = \ut{\allocationSymbol}{\dem}{\agent'}$ for all agents $\agent'$ in $\agents - \agent$, a contradiction since $\allocationSymbol$ is LMMF.
\end{proof}}


\newcommand\numberthis{\addtocounter{equation}{1}\tag{\theequation}}

Theorem~\ref{thm:EF} below shows that any frugal LMMF allocation is
EF.  Bogomolnaia and Moulin~\cite{bogomolnaia2004random} show that any
frugal LMMF allocation is EF when all demands are $0$ or $1$, all
agent endowments are equal, and all object supplies are equal.  To
generalize this result to our setting, the main issue is to handle
fractional demands; Corollaries~\ref{cor:under_demands_rounds}
and~\ref{cor:over_demands_rounds} are useful in this regard.

\begin{theorem}
\label{thm:EF}
Allocation $\allocationSymbol$ is EF. 
\end{theorem}
\punt{\begin{proof}
	
Assume for the sake of contradiction that there are agents $a$ and $a'$ such that agent $a$ envies the allocation of agent $a'$, that is, 
\begin{equation}
	\label{eq:envy_assumption}
	u(\mu,\dem, a) < \sum_{\object\in \objects} \min\left(\frac{\edw(a)}{\edw(a')} \allocationSymbol(a',\object),\dem(a,\object)\right).
\end{equation}
As in Section~\ref{sec:structural}, 
let $k$ denote the value $\num{I}$.
For any $i$ in $[k]$, let $\brkpt_i$, $\agents_i$, and $\objects_i$ denote $\brkpts{I}{i}$, $\agnts{I}{i}$, and $\objs{I}{i}$, respectively.
Let $i$ and $i'$ in $[k]$ be such that agent $\agent$ (resp. $\agent'$) belongs to $\agents_i \setminus \agents_{i-1}$ (resp., $\agents_{i'} \setminus \agents_{i'-1}$).
We consider two cases.

Case~$1$: In this case we have $i'>i$. 
We deduce that 
\begin{align*}
	u(\mu,\dem, a) &= \sum_{\object\in \objects} \min(\allocationSymbol(a,\object),\dem(a,\object))\\
	&= \sum_{\object\in \objects\setminus \objects_i} \min(\allocationSymbol(a,\object),\dem(a,\object)) + \sum_{\object\in \objects_i} \min(\allocationSymbol(a,\object),\dem(a,\object))\\
	&=\dem(a,\objects\setminus \objects_i)+\sum_{\object\in \objects_i} \min(\allocationSymbol(a,\object),\dem(a,\object))\\
	&\ge \dem(a,\objects\setminus \objects_i),  
\end{align*}
where the last equality follows from Corollary~\ref{cor:under_demands_rounds}.
Corollary~\ref{cor:over_demands_rounds} implies that  $\mu(a',\object)=0$ for all objects $\object$ in $\objects_i$. Therefore, 
\begin{align*}
	\sum_{\object\in \objects} \min\left( \frac{\edw(a)}{\edw(a')} \allocationSymbol(a',\object),\dem(a,\object) \right) &=  \sum_{\object\in \objects\setminus \objects_i} \min\left( \frac{\edw(a)}{\edw(a')} \allocationSymbol(a',\object),\dem(a,\object)\right)\\
	&\le \dem(a,\objects\setminus \objects_i).  
\end{align*}
The inequalities derived above imply that 
$$\sum_{\object\in \objects} \min\left( \frac{\edw(a)}{\edw(a')} \allocationSymbol(a',\object),\dem(a,\object)\right) \le u(\mu,\dem,a),$$ contradicting inequality~\eqref{eq:envy_assumption}.

Case $2$: $ i'\le i$. 
Since $\brkpt_1, \dots, \brkpt_k$ is an increasing sequence, $\brkpt_{i'}\le \brkpt_{i}$. 
Lemma~\ref{lem:allocation} implies that $\mu(\agent', \objects)= \edw(a')\brkpt_{i'}$. Thus
\begin{align*}
	\sum_{\object\in \objects} \min\left( \frac{\edw(a)}{\edw(a')} \allocationSymbol(a',\object),\dem(a,\object) \right) &\le \frac{\edw(a)}{\edw(a')} \allocationSymbol(a',\objects) \\
	&=\ \edw(\agent) \brkpt_{i'} \\
	&\le\ \edw(\agent) \brkpt_i \\
	&=\ \ut{\allocationSymbol}{\dem}{\agent},
\end{align*}
where the first and second equalities follow from Lemma~\ref{lem:allocation}.
This inequality contradicts inequality~\eqref{eq:envy_assumption}.
\end{proof}
}

Theorem~\ref{thm:SI} below shows that any frugal LMMF allocation is
$\frac12$-SI.  Lemma~\ref{lem:not_half_SI} in Section~\ref{sec:imp}
implies that no frugal LMMF mechanism is $\CoeffSI$-SI for any $\CoeffSI > \frac12$.

\begin{theorem}
\label{thm:SI}
Allocation $\allocationSymbol$ is $\frac12$-SI.
\end{theorem}
\punt{\begin{proof}
    Let $\agent$ be an agent in $\agents$ and let $\SI(\agent)$ denote $\sum_{\object \in \objects} \min\left( \frac{\edw(\agent)}{\edw(\agents)}\supp(\object), \dem(\agent, \object)\right)$.
    We need to show that $\ut{\allocationSymbol}{\dem}{\agent} \ge \SI(\agent) / 2$.
    
    As in Section~\ref{sec:structural}, let $\Lambda$ and $k$  denote the breakpoint function $\Lambda_I$ and the value $\num{I}$, respectively.
    For any $i$ in $[k]$, let $\brkpt_i$, $\agents_i$, and $\objects_i$ denote $\brkpts{I}{i}$, $\agnts{I}{i}$, and $\objs{I}{i}$, respectively.
    Let $i$ in $[k]$ be such that agent $\agent$ belongs to $\agents_i \setminus \agents_{i-1}$.
    Thus $\brkfn{}{\agent} = \brkpt_i$.
    
    Lemma~\ref{lem:allocation} implies that $\ut{\allocationSymbol}{\dem}{\agent} = \edw(\agent) \brkpt_i$. 
    We have
    $$\SI(a)  =\sum_{\object \in \objects} \min\left( \frac{\edw(\agent)}{\edw(\agents)}\supp(\object), \dem(\agent, \object)\right) \le \frac{\edw(a)}{\edw(A)}\supp(\objects_i) +  \dem(a,\objects\setminus \objects_i).$$ 
    Thus, to prove that $\ut{\allocationSymbol}{\dem}{\agent} \ge \SI(a)/2$, it suffices to prove that
    $u(\mu,\dem, a)\ge \dem(a,\objects\setminus \objects_i)$ and $\ut{\allocationSymbol}{\dem}{\agent} \ge \frac{\edw(\agent)}{\edw(\agents)}\supp(\objects_i)$.
    Observe that $u(\mu,\dem, a)= \mu(a,\objects) \ge \mu(a,\objects\setminus \objects_i) = \dem(a,\objects\setminus \objects_i)$, where the first equality follows from Lemma~\ref{lem:allocation} and the last equality follows from Corollary~\ref{cor:under_demands_rounds}. Thus $u(\mu,\dem, a)\ge \dem(a,\objects\setminus \objects_i)$. 
    
    It remains to prove that $\ut{\allocationSymbol}{\dem}{\agent} \ge \frac{\edw(\agent)}{\edw(\agents)}\supp(\objects_i)$.
    Since $u(\mu,\dem ,a) = \edw(a) \brkpt_i$, it suffices to prove that $\brkpt_i\ge \supp(\objects_i)/\edw(\agents)$. 
    Note that  $\edw(\agents_i) \brkpt_i \ge \sum_{j\in [i]} \edw(\agents_j \setminus \agents_{j-1})\brkpt_j = \supp(\objects_i) + \dem(\agents_i, \objects \setminus \objects_i) \ge \supp(\objects_i)$, where the first inequality follows since $\brkpt_1, \dots, \brkpt_k$ is an increasing sequence and the equality follows from Lemma~\ref{lem:leeq}. 
    Since $\edw(\agents_i) \brkpt_i \ge \supp(\objects_i)$, we find that $\brkpt_i \ge \supp(\objects_i)/\edw(\agents_i) \ge \supp(\objects_i)/\edw(\agents)$, 
    where the last inequality holds because  $\agents_i$ is a subset of $\agents$ and hence $\edw(\agents_i)\le \edw(\agents)$.
\end{proof}
    
}

\subsection{Game-Theoretic Properties of Frugal LMMF Mechanisms} \label{sec:GSP}

In this section, we establish that any frugal LMMF mechanism is RM
  (Theorem~\ref{thm:RM}), PM (Theorem~\ref{thm:PM}), and GSP
  (Theorem~\ref{thm:GSP}).
Bogomolnaia and Moulin~\cite{bogomolnaia2004random} prove these
  properties for the special case where the demands and supplies are all $0$-$1$.
  To handle arbitrary fractional demands and supplies, we employ a similar high-level proof framework, with some additional low-level details. 
  We now
  sketch our RM proof, which proceeds by contradiction.  (See
  Appendix~\ref{sec:RM} for the full proof.) Let
  $I=(\agents,\objects,\edw,\supp,\dem)$ and
  $I'=(\agents,\objects,\edw,\supp',\dem)$ denote $\oafd$ instances
  such that $\supp(\object)\le \supp'(\object)$ for all objects $b$ in
  $\objects$, let $\mu$ belong to $\LMM(I)$, and let $\mu'$ belong to
  $\LMM(I')$.  We first use Lemmas~\ref{lem:allocation}
  and~\ref{lem:frozenAgents_at_diff_iter} to prove that for any agents
  $a$ and $a'$, if there is an object $\object$ such that
  $\mu'(a,\object)< \mu(a,\object)$ and
  $\mu'(a',\object)>\mu(a',\object)$, then 
  $\brkfn{I'}{a'}\leq \brkfn{I'}{a}$ and $\brkfn{I}{a}\leq\brkfn{I}{a'}$.
  Next, to derive a contradiction, we consider the set $A'$ of
  agents who suffer a loss from switching $I$ to $I'$.  A
  straightforward counting argument shows that there is an agent $a'$
  in $A\setminus A'$ and an object $b'$ such that
  $\mu'(a,\object')< \mu(a,\object')$ and
  $\mu'(a',\object')>\mu(a',\object')$, implying that 
    $\brkfn{I'}{a'}\leq \brkfn{I'}{a}$ and $\brkfn{I}{a}\leq\brkfn{I}{a'}$.
  Since agent $a$ suffers a loss from switching $I$ to
  $I'$, we conclude that $a'$ also suffers a loss from switching $I$
  to $I'$. This is a contradiction since $a'$ belongs to $A\setminus A'$ and
  $A'$ contains all agents who suffer a loss from switching $I$ to
  $I'$.
	
  With regard to establishing the PM and GSP properties, a key
  difference between our setting and that of Bogomolnaia and Moulin may be
  illustrated by considering the set of agents receiving the minimum
  utility, i.e., $\agnts{I}{1}$.  When all of the demands are either
  $0$ or $1$, the objects demanded by agents in $\agnts{I}{1}$ are not
  available to the remaining agents, and hence there is a clean
  partitioning of the objects into the subset demanded by agents in
  $\agnts{I}{1}$ and the remaining objects. However, in our setting,
  the objects fractionally demanded by agents in $\agnts{I}{1}$ may
  still be partly available for the remaining agents, which allows for
  a more complicated interplay between the two subproblems. We use
  some new ideas to cope with this added complexity.  For example, in
  the main case (Case~2) of our proof of Theorem~\ref{thm:PM} and in
  the main case (Case~4) of our proof of Theorem~\ref{thm:GSP}, we
  find it convenient to leverage the RM property established in
  Theorem~\ref{thm:RM}.  
  We remark that Bogomolnaia and Moulin do not
  use RM to establish PM or GSP.

\section{Impossibility Results}
\label{sec:imp}

In this section, we show that fairness and SI are incompatible properties.
Lemma~\ref{lem:not_half_SI} below establishes that for any $\CoeffSI > \frac12$, no $\oafd$ mechanism is $\CoeffSI$-SI and MMF.
As mentioned in Section~\ref{sec:prelims}, MMF is a weaker notion of fairness than LMMF.
Thus for any $\CoeffSI > \frac12$, no $\oafd$ mechanism is $\CoeffSI$-SI and LMMF.
\begin{lemma}
	\label{lem:not_half_SI}
	For any $\CoeffSI>\frac12$, no $\oafd$ mechanism is $\CoeffSI$-SI and MMF.
\end{lemma}
\punt{\begin{proof}
	Let $M$ be an MMF $\oafd$ mechanism.  Consider an $\oafd$
        instance $I = (\agents, \objects, \edw, \supp, \dem)$ with $n$
        agents $a_1,\dots,a_n$, each with endowment $1$, and two
        objects $b_1$ and $b_2$, each with supply $n$, and where
        $\dem(a_1, b_1) = \dem(a_1, b_2) = 1$, and $\dem(a, b_1) = 2$
        and $\dem(a, b_2) = 0$ for all agents $a$ in $A - a_1$.
        Mechanism $M$ gives a utility of $1+1/n$ to each agent in $A$.
        If agent $a_1$ is allocated an $\edw(a_1)/\edw(\agents) = 1/n$
        fraction of each object, then $a_1$ achieves utility $2$.
        Hence $M$ is at most $\frac12\left(1 + \frac{1}{n}\right)$-SI.
        Let $\CoeffSI$ be any value greater than $\frac12$.  By choosing
        a sufficiently large $n$, we deduce that $M$ is not
        $\CoeffSI$-SI.
\end{proof}
}


Since no mechanism can be MMF and SI, we consider the following natural relaxation: mechanisms that are MMF subject to being SI.
Formally, for any $\oafd$ instance $I = (\agents, \objects, \edw, \supp, \dem)$, we say that an allocation $\mu$ in $\allocs{I}$ is MMF-SI if $\mu $ maximizes $\min_{\agent \in \agents} \ut{\allocationSymbol'}{\dem}{\agent}/\edw(\agent)$ over all $\allocationSymbol'$ in $\allocs{I}$ such that 
$\allocationSymbol'$ is SI.
An $\oafd$ mechanism $M$ is MMF-SI if 
for any $\oafd$ instance $I$, every allocation in $M(I)$ is MMF-SI.
Lemma~\ref{lem:no_maximin_SI} below shows that the SP and MMF-SI properties are incompatible.
\begin{lemma}
\label{lem:no_maximin_SI}
	No $\oafd$ mechanism is SP and MMF-SI.
\end{lemma}
\punt{\begin{proof}
	Let $M$ be an MMF-SI $\oafd$ mechanism.
	Consider an $\oafd$ instance $I = (\agents, \objects, \edw, \supp, \dem)$ with three agents $a_1$, $a_2$, and $a_3$, each with endowment $1$, and two objects $b_1$ and $b_2$, each with supply $6$, and where $\dem(a_1, b_1) = 3$, $\dem(a_1, b_2) = 1$, $\dem(a_2, b_1) = \dem(a_3, b_1) = 0$, and $\dem(a_2, b_2) = \dem(a_3, b_2) = 3$.
	Let $\allocationSymbol$ belong to $M(I)$.
	It is easy to verify that $\ut{\allocationSymbol}{\dem}{a_1} = 3$.
	Let $\dem'$ denote $(\dem_{A-a_1}, \dem'')$, where $\dem''$ belongs to $\dems{\{a_1\}}{\objects}$, $\dem''(a_1, b_1) = 3$, and $\dem''(a_1, b_2) = 2$.
	Let $I'$ denote the $\oafd$ instance $(\agents, \objects, \edw, \supp, \dem')$ and let $\allocationSymbol'$ belong to $M(I')$.
	It is easy to verify that $\ut{\allocationSymbol'}{\dem'}{a_1} = 4$.
	We conclude that mechanism $M$ is not SP.
\end{proof}}

\section{Concluding Remarks}

In this paper, we introduced the $\oafd$ problem and we presented a lexicographically maximin fair $\oafd$ mechanism that enjoys a number of desirable game-theoretic properties: GSP, $1/2$-SI, NW, EF,  PE, RM, and PM.
We also showed that no maximin fair mechanism can be $\CoeffSI$-SI for any $\CoeffSI > 1/2$.
Further, we showed that no MMF-SI mechanism is SP.

We briefly mention some possible directions for future research.
First, 
we have shown that our mechanism is 
$\CoeffSI$-SI for $\CoeffSI = 1/2$, but on most real world instances it might achieve $\CoeffSI$-SI for a significantly higher value of $\CoeffSI$.
It would be interesting to benchmark our mechanism on real data.
Second, our work assumes perfect knowledge of future demands.
It would be interesting to develop mechanisms whose performance degrades gracefully as the knowledge of future demands becomes more unreliable.
Finally, we have studied lexicographic maximin fairness in this paper. It would also be interesting to study other notions of fairness.

\bibliography{refs}
\newpage

\appendix


\section{Lexicographic Flow}
\label{apd:lex_flow}
The problem of computing a maximum flow in a flow network has been extensively studied. 
We describe this problem and discuss a parameterized version of the problem. We utilize parametric maximum flow to propose an efficient implementation of our algorithm.

A flow network is a directed graph $G = (V, \edges)$ with the vertex set 
$V$, the edge set $\edges$, having a source vertex $s$, a sink vertex $t$, and a non-negative capacity $\capt(\edge)$ for each edge $\edge$ in $\edges$.
A function $f: E \to R_{\ge 0}$ is said to be a flow if 
$f(\edge) \le \capt(\edge)$ (the capacity constraint for edge $e$) holds for each edge $e$ in $E$ and 
$\sum_{(u,v) \in \edges} f(u,v) = \sum_{(v,u) \in \edges} f(v,u)$ (the flow conservation constraint for vertex $v$) holds for each vertex $v$ in $V \setminus \{s, t\}$.
The value of flow $f$ may be defined as the net flow out of the source $s$.
The goal of the maximum flow problem is to determine a flow of maximum value~\cite{ford1962}.

A cut of a flow network $G = (V, E)$ is a partition $(S, \overline{S})$ of $V$ such that $s$ belongs to $S$ and $t$ belongs to $\overline{S}$.
The capacity of a cut $(S, \overline{S})$ is defined as the total capacity of all edges going from some vertex in $S$ to some vertex in $\overline{S}$.
A minimum cut is a cut of minimum capacity.
The famous max-flow min-cut theorem states that in any flow network, the value of a maximum flow is equal to the capacity of a minimum cut.
A standard result in network flow theory states that there is a minimum cut $(S, \overline{S})$ such that $S$ contains $S'$ for all minimum capacity cuts $(S', \overline{S'})$.
We refer to this minimum cut $(S, \overline{S})$ as the source-heavy minimum cut.

In a parametric flow network, each edge capacity is a function of a parameter $\lambda$.
In this paper, we restrict our attention to parametric networks where each edge leaving $s$ has a capacity proportional to $\lambda$ and all other edge capacities are independent of $\lambda$.
Parametric networks have been widely studied; we refer readers to~\cite{gallo1989} for more general settings and other results.
For any parametric flow network $G$, we let $G(\lambda)$ denote the flow network associated with a particular value of $\lambda$.

In a parametric flow network, the capacity of the minimum cut changes as the value of $\lambda$ changes.
We let the minimum cut capacity function $\kappa(\lambda)$ denote the capacity of the minimum cut as a function of the parameter $\lambda$. 
It is well known that $\kappa(\lambda)$ is a non-decreasing, concave, and piecewise-linear function with at most $|V|-2$ breakpoints, where a breakpoint is a value of $\lambda$ at which the slope of $\kappa(\lambda)$ changes~\cite{eisner1976, stone1978}. 
Each of the $|V|-1$ or fewer line segments that form the graph of $\kappa(\lambda)$ corresponds to a cut.
Notice that $\kappa(0) = 0$. 
As the value of $\lambda$ increases, the vertices in $V \setminus \{s,t\}$ move from the sink side to the source side of the source-heavy minimum cut. 
For any parametric flow network $G$, 
the breakpoint function $\brkfn{}{v}$ maps 
any given vertex $v$ in $V \setminus \{s,t\}$ to
the breakpoint value of $\lambda$ at which $v$ moves from the sink side to the source side of the source-heavy minimum cut.
The breakpoint function describes the sequence of cuts associated with $\kappa(\lambda)$ \cite{stone1978}.
 
We now define the notion of a lexicographic flow. 
Assume that 
the edges leaving $s$ reach the vertices $\{v_1, \dots, v_k\}$, and that $t$ does not belong this set. 
Let the capacity of the edge $(s, v_i)$ be $w_i\lambda$. 
For a flow $f$ in $G(\infty)$, let $\theta(G, f)$ denote the length-$k$ vector whose $j$th component is the $j$th smallest $f(s,v_i)/w_i$, for $i$ in $[k]$.
A lexicographic flow $f$ of $G$ is a maximum flow $f$ in $G(\infty)$ that is lexicographically at least  $\theta(G,f')$ for all maximum flows $f'$ in $G(\infty)$.

The lexicographic flow problem has been studied by Megiddo~\cite{megiddo1974,megiddo1977}.
An efficient algorithm for this problem was proposed by Gallo et al.\ \cite{gallo1989}.
We describe this algorithm here.
First, find all the breakpoints of $\kappa(\lambda)$ for a given parametric flow network $G$. 
Also, determine the breakpoint $\lambda(v_i)$ for each vertex $v_i$ in $\{v_1, \dots, v_k\}$ at which $v_i$ moves from the sink side to the source side of the source-heavy minimum cut. 
Let $G'$ be the flow network obtained by setting the capacity of edge $(s, v_i)$ to $w_i\lambda(v_i)$ for each vertex $v_i$ in $\{v_1, \dots, v_k\}$ in $G$.
Any maximum flow $f$ of $G'$ is a lexicographic flow of $G$.
Gallo et al.\ describe an algorithm that computes the breakpoint function and a lexicographic flow in $O(|V||E| \log(|V|^2/|E|))$ time.

\section{Proofs}\label{sec:appendix_missing_proofs}

\subsection{Proof of Lemma~\ref{lem:suboptimal}}

\subsection{Proof of Lemma~\ref{lem:frugal_LMF}}

\subsection{Proof of Lemma~\ref{lem:lambda}}

We now present a proof of Lemma~\ref{lem:lambda}.

\subsection{Proof of Lemma~\ref{lem:allocation}}
\begin{proof}
	The capacity of edge $(\agent, \object)$ in $G$ is $\dem(\agent, \object)$. 
	Hence $\allocationSymbol(\agent, \object) = f(\agent, \object)$ belongs to $[0,\dem(\agent,\object)]$.
	Flow $f$ satisfies $\edw(\agent)\brkfn{}{\agent} = f(s, \agent) = f(\agent, \objects) = \allocationSymbol(\agent, \objects) = \ut{\allocationSymbol}{\dem}{\agent}$.
\end{proof}

\subsection{Proof of Lemma~\ref{lem:leeq}}
\begin{proof}
	We begin by proving the following useful claim.

	Claim~1: Let $i$ belong to $[k]$. 
	Let $\object$ be an object in $\objects_i \setminus \objects_{i-1}$. 
	Then $\suppt{I}{\object} = \supp(\object)$.
	
	To prove Claim~1, observe that the definition of $\objects_i \setminus \objects_{i-1}$ implies that $\dem(\agents_i \setminus \agents_{i-1}, \object) > \capt_i(\object) = \suppt{I}{\object} - \dem(\agents_{i-1}, \object)$.
	Thus $\dem(\agents, \object) \ge \dem(\agents_i, \object) > \suppt{I}{\object}$.
	Since $\dem(\agents, \object) > \suppt{I}{\object}$ and $\suppt{I}{\object} = \min(\supp(\object), \dem(\agents, \object))$, we have $\suppt{I}{\object} = \supp(\object)$.
	This completes the proof of Claim~1.

	Notice that $\sum_{j\in [i]} \edw(\agents_j \setminus \agents_{j-1}) \lambda_j$ is the total flow into $\agents_i$ in $f$.
	The total flow out of $\agents_i$ in $f$ is $f(\agents_i, \objects)$.
	For any agent vertex $\agent$ in $\agents_i$ and any object vertex $\object$ in $\objects \setminus \objects_{i}$, Lemma~\ref{lem:lambda} implies that $f(\agent, \object) = \dem(\agent, \object)$.
	For any object $\object$ in $\objects_i$, Lemma~\ref{lem:lambda} implies that $f(\agents_i, \object) = \suppt{I}{\object}$. 
	Thus $f(\agents_i, \objects) = \suppt{I}{\objects_i} + \dem(\agents_i,\objects\setminus \objects_i)$.
	Since the net flow into $\agents_i$ is $0$, we obtain $\sum_{j\in [i]} \edw(\agents_j \setminus \agents_{j-1}) \lambda_j = \suppt{I}{\objects_i} + \dem(\agents_i,\objects\setminus \objects_i) = \supp(\objects_i) + \dem(\agents_i,\objects\setminus \objects_i)$, where the last equality follows from Claim~1.
\end{proof}

\subsection{Proof of Lemma~\ref{lem:frozenAgents_at_diff_iter}}
\begin{proof}
	Corollary~\ref{cor:over_demands_rounds} and $\allocationSymbol(\agent',\object)>0$
	imply that $\object$ belongs to $\objects \setminus \objects_i$.
	Hence Corollary~\ref{cor:under_demands_rounds}
	implies that $\allocationSymbol(\agent,\object)=\dem(\agent,\object)$.
\end{proof}

\subsection{Proof of Lemma~\ref{lem:minimal_allocation}}

\subsection{Proof of Theorem~\ref{thm:NW}}
\begin{proof}
	The definition of $\allocationSymbol$ implies that $\allocationSymbol$ belongs to $\frugal(I)$.
	Assume for the sake of contradiction that $\allocationSymbol$ is not NW.
	Hence there is an agent $\agent$ in $\agents$ and an object $\object$ in $\objects$ such that $\allocationSymbol(\agent, \object) < \dem(\agent, \object)$ and $\allocationSymbol(\agents, \object) < \supp(\object)$.
	Let $\allocationSymbol'$ be the allocation in $\allocs{I}$ such that $\allocationSymbol'(\agent, \object) = \min(\dem(\agent, \object), \supp(\object) - \allocationSymbol(\agents - \agent, \object))> \allocationSymbol(\agent, \object)$, and $\allocationSymbol'(\agent', \object') = \allocationSymbol(\agent', \object')$ for all $(\agent', \object')$ in $\agents \times \objects - (\agent, \object)$.
	Thus $\ut{\allocationSymbol'}{\dem}{\agent} > \ut{\allocationSymbol}{\dem}{\agent}$ and $\ut{\allocationSymbol'}{\dem}{\agent'} = \ut{\allocationSymbol}{\dem}{\agent'}$ for all agents $\agent'$ in $\agents - \agent$, a contradiction since $\allocationSymbol$ is LMMF.
\end{proof}

\subsection{Proof of Theorem~\ref{thm:EF}}

\subsection{Proof of Theorem~\ref{thm:SI}}

\subsection{Game-Theoretic Properties of Frugal LMMF Mechanisms}
\label{sec:RM}
\newcommand{\smallestI}{j^*}

\begin{theorem}\label{thm:RM}
Any frugal LMMF mechanism is RM.
\end{theorem}
\punt{\begin{proof}
    The definition of mechanism $\LMM$ implies that it is sufficient to show that $\LMM$ is RM.
    Let $I=(\agents,\objects,\edw,\supp,\dem)$ and
    $I'=(\agents,\objects,\edw,\supp',\dem)$ denote $\oafd$ instances such
    that $\supp(\object)\le \supp'(\object)$ for all objects $b$ in
    $\objects$, let $\mu$ belong to $\LMM(I)$, and let $\mu'$ belong to $\LMM(I')$.
    We need to prove that $u(\mu,\dem,a) \le u(\mu',\dem,a)$ for all agents $a$ in
    $\agents$.
    Let $\Lambda$ and $\Lambda'$ denote the breakpoint functions for $\Lambda_I$
    and $\Lambda_{I'}$, respectively.  We begin by proving the following useful
    claim.
    
    Claim~1: Let $a$ and $a'$ be agents in $\agents$ and let $b$ be an
    object in $\objects$ such that $\mu'(a,\object)< \mu(a,\object)$ and
    $\mu'(a',\object)>\mu(a',\object)$. Then $\brkfn{}{a}\leq\brkfn{}{a'}$ and
    $\Lambda'(a')\leq \Lambda'(a)$.
    
    To prove Claim~1, first observe that
    $0\leq\mu'(a,\object)<\mu(a,\object)\leq\dem(a,\object)$ and
    $0\leq\mu(a',\object)<\mu'(a',\object)\leq\dem(a',\object)$ by
    Lemma~\ref{lem:allocation}. Since $\mu(a,\object)>0$ and
    $\mu(a',\object)<\dem(a',\object)$, Lemma~\ref{lem:frozenAgents_at_diff_iter}
    implies that $\brkfn{}{a}\leq\brkfn{}{a'}$. Similarly, since
    $\mu'(a',\object)>0$ and $\mu'(a,\object)<\dem(a,\object)$,
    Lemma~\ref{lem:frozenAgents_at_diff_iter} implies that
    $\Lambda'(a')\leq\Lambda'(a)$. This completes the proof of Claim~1.
    
    Let $\agents'$ denote $\{a\in\agents\mid u(\mu,\dem, a)>u(\mu',\dem,
    a)\}$. To establish the lemma, we need to prove that $\agents'$ is
    empty. Assume for the sake of contradiction that $\agents'$ is
    nonempty. Let $\lambda^*$ denote $\min_{a\in A'}\Lambda'(a)$, and let
    $\agents''$ denote $\{a\in\agents'\mid\Lambda'(a)=\lambda^*\}$; thus
    $\agents''$ is nonempty.
    
    Let $\objects'$ denote
    $\{b\in\objects\mid\mu(\agents'',\object)>\mu'(\agents'',\object)\}$. The set
    $\objects'$ is nonempty since $\agents''$ is a nonempty subset of
    $\agents'$.  Let $\object$ denote an object in $\objects'$.
    
    Let $\agents'''$ denote $\{a\in\agents''\mid\mu(a,\object)>\mu'(a,\object)\}$. The
    set $\agents'''$ is nonempty since
    $\mu(\agents'',\object)>\mu'(\agents'',\object)$. Let $a$ denote an agent in
    $\agents'''$. Since $u(\mu,\dem, a) =\edw(a)\brkfn{}{a}$ and
    $u(\mu',\dem, a)=\edw(a)\Lambda'(a)$ by
    Lemma~\ref{lem:allocation}, and since $a$ belongs to $\agents'$, we
    deduce that $\brkfn{}{a}>\Lambda'(a)=\lambda^*$.
    
    Since $\mu(\agents'',\object)>\mu'(\agents'',\object)$ and Theorem~\ref{thm:NW}
    implies that $\mu(\agents,\object)=\mu'(\agents,\object)$, we deduce that there is
    an agent in $\agents\setminus\agents''$, call it $a'$, such that
    $\mu(a',\object)<\mu'(a',\object)$.
    
    Since $\mu(a,\object)>\mu'(a,\object)$ and $\mu(a',\object)<\mu'(a',\object)$, Claim~1 implies
    that $\brkfn{}{a}\leq\brkfn{}{a'}$ and $\Lambda'(a')\leq
    \Lambda'(a)=\lambda^*$. Since $\brkfn{}{a}>\lambda^*$, we have
    $\Lambda'(a')\leq\lambda^*<\brkfn{}{a}\leq\brkfn{}{a'}$.
    Thus $a'$ belongs to $\agents'$, and hence the definition of
    $\lambda^*$ implies $\Lambda'(a')\geq\lambda^*$. Since
    $\Lambda'(a')\leq\lambda^*$, we conclude that
    $\Lambda'(a')=\lambda^*$.  Since $a'$ belongs to $\agents'$ and
    $\Lambda'(a')=\lambda^*$, we deduce that $a'$ belongs to $\agents''$, a
    contradiction. 
\end{proof}
    
}


\newcommand{\specificAgent}{a^\dagger}

\begin{theorem}\label{thm:PM}
	Any frugal LMMF mechanism is PM.
\end{theorem}

\punt{\begin{proof}
	By the definition of mechanism $\LMM$, it is sufficient to
        show that $\LMM$ is PM.  Let $P(k)$ denote the predicate ``for
        any $\oafd$ instances $I = (\agents, \objects, \edw, \supp,
        \dem)$ and $I' = (A', \objects, \edw', \supp,\dem_{ A'})$ such
        that $|\agents|=k$ and $I'$ belongs to $\shrink(I)$, any
        allocations $\allocationSymbol$ in $\LMM(I)$ and
        $\allocationSymbol'$ in $\LMM(I')$, and any agent $\agent$ in
        $ A'$ such that $\edw'(\agent) = \edw(\agent)$, we have
        $u(\mu,\dem,a) \le u(\mu',\dem,a)$.''  We prove by induction
        that $P(k)$ holds for all $k\ge 0$, which implies that the
        theorem holds.
	
	Base case: It is easy to see that $P(0)$ holds.
	
	Induction step: Let $k$ be a positive integer and assume that $P(i)$
	holds for $0\leq i<k$. We need to prove that $P(k)$ holds.  Let
	$I=(\agents,\objects,\edw,\supp,\dem)$ and
	$I'=(A',\objects,\edw',\supp,\dem_{ A'})$ be $\oafd$ instances such that
	$|\agents|=k$ and $I'$ belongs to $\shrink(I)$.  Let allocation
	$\allocationSymbol$ (resp., $\allocationSymbol'$) belong to $\LMM(I)$
	(resp., $\LMM(I')$).  Let $\specificAgent$ be an agent in $A'$ such
	that $\edw'(\specificAgent) = \edw(\specificAgent)$.  We need to prove
	that $u(\mu,\dem,\specificAgent) \le u(\mu',\dem,\specificAgent)$.
	Let $\brkpt_1$ (resp., $\brkpt'_1$) denote $\brkpts{I}{1}$ (resp.,
	$\brkpts{I'}{1}$), and let $\agents_1$ (resp., $\agents'_1$) denote
	$\agnts{I}{1}$ (resp., $\agnts{I'}{1}$).  We consider two cases.
	
	Case~1: $\specificAgent \in A_1$.  Since $\specificAgent$ belongs to
	$A_1$, Lemma~\ref{lem:allocation} implies that
	$u(\mu,\dem,\specificAgent) = \lambda_1 \edw(\specificAgent)$.  The
	definition of $\lambda_1'$ implies that
	$u(\mu',\dem,\specificAgent)\ge \lambda_1' \edw(\specificAgent)$.
	Since $u(\mu,\dem,\specificAgent) = \lambda_1 \edw(\specificAgent)$
	and $u(\mu',\dem,\specificAgent)\ge \lambda_1' \edw(\specificAgent)$,
	it is sufficient to prove that $\lambda_1'\ge \lambda_1$.  Let
	$\mu_{\agents'}$ denote the restriction of $\mu$ to $\agents'$; thus
	$\mu_{A'}$ belongs to $\frugal(I')$.  Since $\mu'$ belongs to
	$\LMMF(I')$, we deduce that $\mathbf{u}(I', \allocationSymbol')$ is
	lexicographically at least $\mathbf{u}(I', \mu_{A'})$.
	Hence $\brkpt_1' \ge \brkpt_1$, as required.
	
	Case~2: $\specificAgent \in A'\setminus A_1$.  Let
	$\hat{I}=(\hat{\agents},\objects,\hat{\edw},\hat{\supp},\hat{\dem})$
	and
	$\hat{I}'=(\hat{\agents}',\objects,\hat{\edw}',\hat{\supp}',\hat{\dem}')$
	denote the OAFD instances such that $\hat{I} =
	\newsub{I}{\mu}{\agents_1\cup (A\setminus A')}$ and $\hat{I}' =
	\newsub{I'}{\mu'}{\agents_1\cap A'}$. Notice that $\hat{\agents}=
	A\setminus (\agents_1\cup (A\setminus A')) = A'\setminus A_1 =
	\hat{\agents}'$ and $\hat{\dem} = \dem_{\hat{A}} = \dem_{\hat{A}'} =
	\hat{\dem}'$. Moreover, the case assumption implies that
	$\specificAgent$ belongs to $A'\setminus A_1 = \hat{\agents}$.  Let
	$\hat{\mu}$ and $\hat{\mu}'$ be allocations in $\LMM(\hat{I})$ and
	$\LMM(\hat{I}')$, respectively.  By Lemma~\ref{lem:suboptimal}, it is
	sufficient to prove that $u(\hat{\mu},\dem,\specificAgent) \le
	u(\hat{\mu}',\dem,\specificAgent)$.
	
	Since $\hat{\edw}$ (resp., $\hat{\edw}'$) is the restriction
        of $\edw$ (resp., $\edw'$) to $\hat{A}$, we have
        $\hat{\edw}'(a)\le \hat{\edw}(a)$ for all agents $a$ in
        $\hat{A}$.  Let $\hat{I}^*$ denote the $\oafd$ instance
        $(\hat{\agents},\objects,\hat{\edw},\hat{\supp}',\hat{\dem})$,
        which belongs to $\shrink(\hat{I}')$.  Let $\hat{\mu}^*$
        denote an allocation in $\LMM(\hat{I}^*)$.  The induction
        hypothesis implies that $u(\hat{\mu}',\dem,a) \ge
        u(\hat{\mu}^*,\dem,a)$ for all agents $a$ in $\hat{A}$ such
        that $\hat{\edw}(a)=\hat{\edw}'(a)$. Since $\specificAgent$
        belongs to $\hat{A}$ and $\hat{\edw}'(\specificAgent) =
        \hat{\edw}(\specificAgent)$, we deduce that
        $u(\hat{\mu}^*,\dem,\specificAgent) \le
        u(\hat{\mu}',\dem,\specificAgent)$. Below we complete the
        proof by showing that $u(\hat{\mu},\dem,\specificAgent) \le
        u(\hat{\mu}^*,\dem,\specificAgent)$.
	
	Let $\object$ be an object in $\objects$.  We have
	\[
	\mu(\agents_1\cup (A\setminus A'),\object)
	\ge \mu(\agents_1,\object)
	= \supp_I(\object)
	= \min\left(\supp(\object), \sum_{a\in A} \dem(a,\object) \right)
	\ge \mu'(A', \object),
	\]
	where the first equality holds by Lemma~\ref{lem:lambda}, the
        second equality holds by the definition of $\supp_I(\object)$,
        and the second inequality holds because $\mu'$ belongs to
        $\frugal(I')$.  Therefore, $\mu(\agents_1\cup (A\setminus
        A'),\object) \ge \mu'(A', \object) \ge \mu'(A'\cap A_1,
        \object)$. Since $\hat{\supp}(\object) = \supp(\object) -
        \mu(\agents_1\cup (A\setminus A'),\object)$ and
        $\hat{\supp}'(\object) = \supp(\object) - \mu'(A'\cap
        A_1,\object)$, we deduce that
        $\hat{\supp}(\object)\le\hat{\supp}'(\object)$. Hence
        Theorem~\ref{thm:RM} implies that
        $u(\hat{\mu},\dem,\specificAgent) \le
        u(\hat{\mu}^*,\dem,\specificAgent)$, as required.
\end{proof}
}




\begin{theorem}
\label{thm:GSP}
Any frugal LMMF mechanism is GSP. 
\end{theorem}
\punt{\begin{proof}
The definition of mechanism $\LMM$ implies that it is sufficient to show that mechanism $\LMM$ is GSP.
For any OAFD instances $I=(\agents,\objects,\edw,\supp,\dem)$ and
$I'=(\agents,\objects,\edw,\supp,\dem')$, any subset $\agents'$ of
$\agents$ such that
$\dem_{\agents\setminus\agents'}=\dem'_{\agents\setminus\agents '}$,
and any allocation $\mu'$ in $\LMM(I')$, we define
$(I,I',\agents',\mu')$ as a manipulation.  For any manipulation
$\Manip=(I,I',\agents',\mu')$ where
$I=(\agents,\objects,\edw,\supp,\dem)$, we define the set of winning
agents, denoted $\Winners{\Manip}$, as $\{a\in \agents\mid
u(\mu',\dem,a)>\edw(a)\brkfn{I}{a}\}$.  Similarly, we define the
set $\Losers{\Manip}$ of losing agents as $\{a\in\agents\mid
u(\mu',\dem,a)<\edw(a)\brkfn{I}{a}\}$.  Remark:
Lemma~\ref{lem:allocation} implies that
$u(\mu,\dem,a)=\edw(a)\brkfn{I}{a}$ for all 
allocations 
$\mu$ in
$\LMM(I)$ and all agents $a$ in $\agents$.

Let $P(k)$ denote the predicate ``for any manipulation
$\Manip=(I,I',\agents',\mu')$ where
$I=(\agents,\objects,\edw,\supp,\dem)$, $|\agents|=k$, and
$\Winners{\Manip}\cap\agents'\not=\emptyset$, we have
$\Losers{\Manip}\cap\agents'\not=\emptyset$.'' Below we prove by
induction on $k$ that $P(k)$ holds for all $k\geq 0$; the claim of the
theorem follows immediately.

It is easy to see that $P(0)$ holds. Let $k$ be a positive integer and
assume that $P(i)$ holds for $0\leq i<k$. We need to prove that $P(k)$
holds.  Let $\Manip=(I,I',\agents',\mu')$ be a manipulation where
$I=(\agents, \objects,\edw,\supp,\dem)$, $I'=(\agents,
\objects,\edw,\supp,\dem')$, $|\agents|=k$, and
$\Winners{\Manip}\cap\agents'\not=\emptyset$.  We need to prove that
$\Losers{\Manip}\cap\agents'\not=\emptyset$.  Let $\brkpt_1$ (resp.,
$\brkpt'_1$) denote $\brkpts{I}{1}$ (resp.,
$\brkpts{I'}{1}$), and let $\agents_1$ (resp., $\agents'_1$) denote
$\agnts{I}{1}$ (resp., $\agnts{I'}{1}$).  We consider four cases.

Case~1: $\brkpt'_1<\brkpt_1$.  Lemma~\ref{lem:lambda} implies that
$\brkpt'_1$ is equal to $\capty{I'}{\agents'_1}/\edw(\agents'_1)$ and
$\brkpt_1$ is equal to  $\min_{X\subseteq A}\capty{I}{X}/\edw(X)$. Thus
\[
\capty{I'}{\agents'_1}/\edw(\agents'_1)
<
\min_{X\subseteq\agents}\capty{I}{X}/\edw(X)
\leq
\capty{I}{\agents'_1}/\edw(\agents'_1),
\]
where the first inequality follows from the case assumption and the
second inequality follows from $\agents'_1\subseteq\agents$.
Multiplying by $\edw(\agents'_1)$, we obtain
$\capty{I'}{\agents'_1}<\capty{I}{\agents'_1}$.  If
$A'\cap\agents_1'=\emptyset$, then $\dem'_a=\dem_a$ for all agents $a$
in $\agents'_1$ and hence
$\capty{I'}{\agents'_1}=\capty{I}{\agents'_1}$, a contradiction.  It
remains to consider the case where
$A'\cap\agents'_1\not=\emptyset$. Let $a$ belong to
$A'\cap\agents'_1$. Thus
\[
u(\mu',\dem, a)
\leq\mu'(a,\objects)
=u(\mu',\dem',a)
=\edw(a)\brkpt'_1
<\edw(a)\brkpt_1,
\]
where the two equalities follows from
Lemma~\ref{lem:allocation}. Hence $a$ is in
$\Losers{\Manip}\cap\agents'$.

Case~2: $\brkpt'_1\geq\brkpt_1$ and
$\Losers{\Manip}\cap\agents_1\not=\emptyset$.  Let $a$ be an agent in
$\Losers{\Manip}\cap\agents_1$.  If $a$ is in $\agents'$ then
$\Losers{\Manip}\cap\agents'\not=\emptyset$, as required. Thus, in
what follows, we assume that $a$ is not in $\agents'$.  Let $i$ denote
the least integer such that $a$ is in $\agnts{I'}{i}$.
We have
\[
\edw(a)\brkfn{I'}{a}
=u(\mu',\dem', a)
=u(\mu',\dem, a)
<\edw(a)\lambda(a)
=\edw(a)\brkpt_1,
\]
where the first equality holds by Lemma~\ref{lem:allocation}, the
second equality holds since $a$ is not in $\agents'$ and hence
$\dem'_a=\dem_a$, the inequality holds since $a$ is in
$\Losers{\Manip}$, and the third equality holds since $a$ is in
$\agents_1$. Thus $\brkpt'_1\leq\brkfn{I'}{a}<\brkpt_1$,
contradicting the first condition in the case assumption.

Case~3: $\brkpt'_1\geq\brkpt_1$,
$\Losers{\Manip}\cap\agents_1=\emptyset$, and
$\Winners{\Manip}\cap\agents_1\neq\emptyset$. Let $a$ denote an agent
in $\Winners{\Manip}\cap\agents_1$.  Thus $u(\mu',\dem,
a)>\edw(a)\brkfn{I}{a}=\edw(a)\brkpt_1$.  Since
$u(\mu',\dem,\agents_1)\leq\capty{I}{\agents_1}=\edw(\agents_1)\brkpt_1$
by Lemma~\ref{lem:lambda}, we deduce that $u(\mu',\dem,
\agents_1-a)<\edw(\agents_1-a)\brkpt_1$.  Thus there is an agent
$a'$ in $\agents_1-a$ such that
$u(\mu',\dem,a')<\edw(a')\brkpt_1=\edw(a')\brkfn{I}{a'}$. Hence
$\Losers{\Manip}\cap\agents_1\not=\emptyset$, contradicting the second
condition in the case assumption.

Case~4: $\brkpt'_1\geq\brkpt_1$,
$\Losers{\Manip}\cap\agents_1=\emptyset$, and
$\Winners{\Manip}\cap\agents_1=\emptyset$.  Let $\mu$ denote an
allocation in $\LMM(I)$.  Let
$\hat{I}=(\hat{\agents},\objects,\hat{\edw},\hat{\supp},\hat{\dem})$
denote the OAFD instance $\newsub{I}{\mu}{\agents_1}$; thus
$\hat{\agents}=\agents\setminus\agents_1$. Let $\hat{\mu}$ denote the
restriction of $\mu$ to $\hat{\agents}$; Lemma~\ref{lem:suboptimal}
implies that $\hat{\mu}$ is in $\LMM(\hat{I})$.  Let
$I^*=(\hat{\agents},\objects,\hat{\edw},\supp^*,\dem^*)$ denote the
OAFD instance $\newsub{I}{\mu'}{\agents_1}$ and let $\mu^*$ denote the
restriction of $\mu'$ to $\hat{\agents}$; Lemma~\ref{lem:suboptimal}
implies that $\mu^*$ is in $\LMM(I^*)$. Let $\tilde{I}$ denote the
OAFD instance $(\hat{\agents},\objects,\hat{\edw},\supp^*,\hat{\dem})$
and let $\tilde{\mu}$ be in $\LMM(\tilde{I})$.

Claim~1: $\brkfn{I}{a}\geq\brkfn{\tilde{I}}{a}$ holds for all agents
$a$ in $\hat{\agents}$.  The third condition in the case assumption
implies that $u(\mu',\dem,a)\geq\edw(a)\brkpt_1$ for all agents
$a$ in $\agents_1$. Thus Lemma~\ref{lem:minimal_allocation} implies
that $\mu(\agents_1,\object)\leq\mu'(\agents_1,\object)$ for all objects $b$ in
$\objects$. It follows that $\hat{\supp}(\object)\geq\supp^*(\object)$ for all
objects $b$ in $\objects$. Hence
\[
\edw(a)\brkfn{I}{a}
=u(\mu,\dem,a)
=u(\hat{\mu},\hat{\dem},a)
\geq u(\tilde{\mu},\hat{\dem},a)
=\edw(a)\brkfn{\tilde{I}}{a},
\]
where the first and last equalities hold by
Lemma~\ref{lem:allocation}, the second equality holds by the
definition of $\hat{\mu}$, and the inequality holds by
Theorem~\ref{thm:RM}. Dividing by $\edw(a)$ yields the claim.

Claim~2: $u(\mu',\dem,a)=u(\mu^*,\hat{\dem},a)$ for all agents $a$
in $\hat{\agents}$. We have
\[
u(\mu',\dem,a)
=
\sum_{b\in \objects}\min(\mu'(a,\object),\dem(a,\object))
=
\sum_{b\in \objects}\min(\mu^*(a,\object),\hat{\dem}(a,\object))
=
u(\mu^*,\hat{\dem},a),
\]
where the second equality holds by the definition of $\mu^*$. The
claim follows.

Let $\agents''$ denote $\agents'\setminus\agents_1$ and let $\Manip'$
denote the manipulation $(\tilde{I},I^*,\agents'',\mu^*)$.

Claim~3: $\Winners{\Manip'}\cap\agents''\not=\emptyset$.  Since
$\Winners{\Manip}\cap\agents'\not=\emptyset$, the third condition in
the case assumption implies that
$\Winners{\Manip}\cap\agents''\not=\emptyset$. Let $a$ be an agent in
$\Winners{\Manip}\cap\agents''$. 
Thus
\[
u(\mu^*,\hat{\dem},a)
=
u(\mu',\dem,a)
>
\edw(a)\brkfn{I}{a}
\geq
\edw(a)\brkfn{\tilde{I}}{a},
\]
where the equality holds by Claim~2, the first inequality holds
because $a$ is in $\Winners{\Manip}$, and the second inequality holds
by Claim~1. Since $u(\mu^*,\hat{\dem},a)>\brkfn{\tilde{I}}{a}$ and
$a$ is in $\agents''$, the claim holds.

Since $|\hat{\agents}|<k$ and Claim~3 holds, the induction hypothesis
implies that $\Losers{\Manip'}\cap\agents''\not=\emptyset$.  Let $a$
be in $\Losers{\Manip'}\cap\agents''$. Thus $a$ is in
$A''\subseteq\hat{A}$ and
\[
u(\mu',\dem,a)
=
u(\mu^*,\hat{\dem},a)
<
\edw(a)\brkfn{\tilde{I}}{a}
\leq
\edw(a)\brkfn{I}{a},
\]
where the equality holds by Claim~2, the first inequality holds
because $a$ is in $\Losers{\Manip'}$, and the second inequality holds
by Claim~1. Since $u(\mu',\dem,a)<\edw(a)\brkfn{I}{a}$ and $a$
is in $\agents''\subseteq\agents'$, we deduce that $a$ is in
$\Winners{\Manip}\cap\agents'$. 
\end{proof}
}

\subsection{Proof of Lemma \ref{lem:not_half_SI}}
\begin{proof}
	Let $M$ be an MMF $\oafd$ mechanism.  Consider an $\oafd$
        instance $I = (\agents, \objects, \edw, \supp, \dem)$ with $n$
        agents $a_1,\dots,a_n$, each with endowment $1$, and two
        objects $b_1$ and $b_2$, each with supply $n$, and where
        $\dem(a_1, b_1) = \dem(a_1, b_2) = 1$, and $\dem(a, b_1) = 2$
        and $\dem(a, b_2) = 0$ for all agents $a$ in $A - a_1$.
        Mechanism $M$ gives a utility of $1+1/n$ to each agent in $A$.
        If agent $a_1$ is allocated an $\edw(a_1)/\edw(\agents) = 1/n$
        fraction of each object, then $a_1$ achieves utility $2$.
        Hence $M$ is at most $\frac12\left(1 + \frac{1}{n}\right)$-SI.
        Let $\CoeffSI$ be any value greater than $\frac12$.  By choosing
        a sufficiently large $n$, we deduce that $M$ is not
        $\CoeffSI$-SI.
\end{proof}

\subsection{Proof of Lemma \ref{lem:no_maximin_SI}}
\begin{proof}
	Let $M$ be an MMF-SI $\oafd$ mechanism.
	Consider an $\oafd$ instance $I = (\agents, \objects, \edw, \supp, \dem)$ with three agents $a_1$, $a_2$, and $a_3$, each with endowment $1$, and two objects $b_1$ and $b_2$, each with supply $6$, and where $\dem(a_1, b_1) = 3$, $\dem(a_1, b_2) = 1$, $\dem(a_2, b_1) = \dem(a_3, b_1) = 0$, and $\dem(a_2, b_2) = \dem(a_3, b_2) = 3$.
	Let $\allocationSymbol$ belong to $M(I)$.
	It is easy to verify that $\ut{\allocationSymbol}{\dem}{a_1} = 3$.
	Let $\dem'$ denote $(\dem_{A-a_1}, \dem'')$, where $\dem''$ belongs to $\dems{\{a_1\}}{\objects}$, $\dem''(a_1, b_1) = 3$, and $\dem''(a_1, b_2) = 2$.
	Let $I'$ denote the $\oafd$ instance $(\agents, \objects, \edw, \supp, \dem')$ and let $\allocationSymbol'$ belong to $M(I')$.
	It is easy to verify that $\ut{\allocationSymbol'}{\dem'}{a_1} = 4$.
	We conclude that mechanism $M$ is not SP.
\end{proof}

\end{document}